\documentclass[preprint]{elsarticle}

\usepackage[T1]{fontenc}

\usepackage{hyperref}
\usepackage{color}

\usepackage{amsmath}
\usepackage{url}

\usepackage{microtype}
\usepackage{amsmath}
\usepackage{latexsym}
\usepackage{amssymb}
\usepackage{verbatim}
\usepackage{pbox}
\usepackage{enumerate}
\usepackage{multicol}
\usepackage{xcolor}
\usepackage{booktabs}

\usepackage{thmtools}

\usepackage{graphicx}
\usepackage{tikz}
\usetikzlibrary{tikzmark}

\usepackage{txfonts}
{  
   \DeclareSymbolFont{symbolsC}{U}{txsyc}{m}{n}
   \DeclareMathSymbol{\strictif}{\mathrel}{symbolsC}{74}
   \DeclareMathSymbol{\boxright}{\mathrel}{symbolsC}{128}
}

\newtheorem{theorem}{Theorem}
\newtheorem{corollary}[theorem]{Corollary}
\newtheorem{lemma}[theorem]{Lemma}
\newtheorem{example}[theorem]{Example}
\newtheorem{definition}[theorem]{Definition}

\newtheorem{proposition}[theorem]{Proposition}
\newtheorem{remark}[theorem]{Remark}

\newproof{proof}{Proof}    

\newcommand{\dep}[2]{=\hspace{-3pt}({#1};{#2})}

\newcommand{\cf}{\boxright}

\newcommand{\PC}{\PP(\cf)}
\newcommand{\PO}{\PP(\supset)}
\newcommand{\PPs}{\PP(\lor)}
\newcommand{\PP}{\mathcal{P}}

\newcommand{\CO}{\mathcal{CO}}
\newcommand{\COD}{\mathcal{COD}}

\newcommand{\PCO}{\mathcal{PCO}}
\newcommand{\Ps}{\PP(\lor)}
\newcommand{\Psm}{\PP^-(\lor)}
\newcommand{\PCOs}{\PCO(\lor)}

\newcommand{\Al}{\hat \alpha}
\newcommand{\B}{\mathbb{B}}

\newcommand{\F}{\mathcal{F}}
\newcommand{\G}{\mathcal{G}}
\newcommand{\K}{\mathcal{K}}

\newcommand{\La}{\mathcal{L}}

\newcommand{\Q}{\mathbb{Q}}

\newcommand{\pvec}{\overline{\mathrm p}}

\newcommand{\Pvec}{\overline{\mathrm P}}
\newcommand{\Qvec}{\overline{\mathrm Q}}

\newcommand{\dfn}{\mathrel{\mathop:}=}
\newcommand{\dom}{\mathrm{Dom}}
\newcommand{\ran}{\mathrm{Ran}}
\newcommand{\PA}{\mathrm{PA}}
\newcommand{\en}{\mathrm{End}}
\newcommand{\exo}{\mathrm{Exo}}

\newcommand{\Team}{\mathrm{Team}}

\newcommand{\SET}[1]{\mathbf{#1}}

\newcommand{\pindep}{\rotatebox[origin=c]{90}{$\models$}}

\newcommand{\corrf}[1]{#1}

\author[1]{Fausto Barbero\corref{cor1}}%
\ead{fausto.barbero@helsinki.fi}

\author[2]{Jonni Virtema}%
\ead{j.t.virtema@sheffield.ac.uk}

\cortext[cor1]{Corresponding author}

\affiliation[1]{organization={University of Helsinki},
addressline={Yliopistonkatu 3},
postcode={00014},
city={Helsinki},
country={Finland}}

\affiliation[2]{organization={University of Sheffield},
addressline={Western Bank},
city={Sheffield},
postcode={S10 2TN},
country = {United Kingdom}}

\title{Expressivity Landscape for Logics with Probabilistic Interventionist Counterfactuals}

\begin{document}

\begin{abstract}
Causal multiteam semantics is a framework where probabilistic dependencies arising from data and 
causation between variables can be formalized together and studied logically.
We discover complete characterizations of expressivity for several logics that can express probabilistic statements, conditioning and interventionist counterfactuals. 
The results characterize the languages in terms of families of linear inequalities 
and closure conditions  that define the corresponding classes of causal multiteams; we find that the strict tensor disjunction, an operator typical of team semantics but absent from the literature on causation, is needed to capture the full class of all linear inequalities.

The characterizations yield a strict hierarchy of expressive power and some undefinability results.
\end{abstract}


%
%



%
%



\maketitle

\section{Introduction}\label{sec:introduction}
The main approach to the study of empirical data in the 20th century has been that of statistics, which makes use of probabilistic notions such as \emph{correlation} and \emph{conditional (in)dependence} between variables. We follow here another line of study -- going back at least to Sewall Wright \cite{Wri1921} -- insisting that the analysis should not stop at correlations, but instead should yield information about causation among variables (conditional on appropriate scientific assumptions). The methods involved in the analysis of causes and effects have gained in popularity in the last decades, and their mathematics has been vastly developed under the label of \emph{causal inference} (see, e.g., \cite{Pea2000,SpiGlySch1993}). Today the methods of causal inference are heavily utilized, e.g., in epidemiology \cite{HerRob2023}, econometrics \cite{HecVyt2007}, social sciences \cite{MorWin2015} and machine learning \cite{PetJanSch2017,Sch2022}.

One of the next crucial steps in the development of artificial intelligence will be the capability of AI systems to represent and reason about causal knowledge (see, e.g., \cite{PeaMac2018}). 
For the development of AI applications of causal inference, the clarification of the related formal logical theory is vital. It turns out that many concepts involved in the analysis of causes can be reduced to the study of \emph{interventionist counterfactuals} in \emph{causal models}.
Causal models represent causation between variables using so-called \emph{structural equations}, which describe deterministic causal laws that relate the variables to each other. 
In their simplest form, interventionist counterfactuals are expressions of the form
\begin{center}
``if variables $X_1,\dots,X_n$ were set to values $x_1,\dots,x_n$, then $Y$ would take value $y$''.
\end{center}     
Such conditionals are counterfactual (contrary to fact) in that their evaluation forces us to consider an alternative scenario in which the variables $X_1,\dots,X_n$ are subtracted to the laws that currently determine their behaviour, and in which the (possibly new) values taken by such variables are fixed by some external intervention. The causal laws encoded in the model then allow us to find out, computationally, how  all the variables in the system are affected in this alternative scenario. 
Research on logics encompassing interventionist counterfactuals has been active in the past two decades. For example, a number of publications have provided complete axiomatizations for languages of various syntax, and over different classes of models.\footnote{The first such instance was the paper \cite{GalPea1998}, which only considered conjunctions of basic counterfactuals over \emph{recursive} causal models (those which exclude circular causation) and unique-solution models (whose precise definition is besides the needs of this paper). The results of \cite{Hal2000} allowed for free usage of Boolean connectives (except in the antecedents of counterfactuals) and also treateded models with circular causation. \cite{Bri2012} showed how to deal with right-nesting of counterfactuals and Boolean antecedents. \cite{BarSan2020,BarYan2020} axiomatized  languages with an operator for expressing observations (\emph{selective implication}) and with atoms expressing data dependencies among variables. This list of papers is not exhaustive, but it illustrates well the gradual build up to the languages we shall consider in this paper.} 
The papers \cite{Hal2013,Zha2013} drew precise connections with the earlier Stalnaker-Lewis theory of counterfactuals \cite{Sta1968,Lew1973}. 
In \cite{BarGal2022} logics for causal reasoning were studied via translations to first-order logic,
and the articles
\cite{Hal2000,Hal2016,MosIbeIca2022} discuss the complexity of causal and probabilistic languages. 

The classical literature on causal inference does not neatly separate the methods of probability and of causal analysis; many standard concepts in causal inference are expressed by mixing probabilistic and causal concepts. In other words, causal inference uses an array of new notational devices that are not entirely reducible to classical probabilistic reasoning; two significant examples (from \cite{Pea2000}) of these new notations are the conditional $do$ expressions
($Pr(y\mid do(x),z)$)
and Pearl's ``counterfactuals'' 
($Pr(Y_{X = x} \mid Z=z)$). 
We refer the reader to \cite{BarCorIbeIca2022} for a detailed discussion of the meaning and use of these expressions. Roughly speaking, 
they both describe the probability that the variable $Y$ takes value $y$ after intervening to set $X$ to $x$, conditional upon the observation that $Z$ takes value $z$; 
but the two expressions differ subtly in that in the former the conditioning 
is performed in the system modified by the intervention that sets $X$ to $x$, while in the latter expression conditioning 
is 
relative to the pre-intervention system. 
To this regard, we follow the proposal of Barbero and Sandu  \cite{BarSan2018,BarSan2024} to decompose these complex causal-probabilistic expressions in terms of a 
minimal set of logical primitives. In particular, 
 probabilistic conditioning 
 and causal interventions will correspond to two distinct logical conditionals, $\supset$ and $\cf$.


In order to 
make this decomposition possible, one needs to move from causal models to the more general \emph{causal multiteam semantics}, where 
all the needed logical operators are available.
Team semantics is the semantical framework of modern logics of dependence and independence. Introduced by Hodges \cite{Hod1997} and adapted to dependence logic by V\"a\"an\"anen \cite{Vaa2007}, team semantics defines truth in reference to collections of assignments, called \emph{teams}. Team semantics is particularly suitable for the formal analysis of dependencies and independencies in data. Recent developments in the area have broadened the scope of team semantics to cover probabilistic and quantitative notions of dependence and independence. Durand et al. \cite{DurHanKonMeiVir2018,DurHanKonMeiVir2016} introduced multiset and probabilistic variants of team semantics as frameworks for studying probabilistic dependency notions such as conditional independence logically. Further analysis has revealed that definability and complexity of logics in these frameworks are intimately connected to definability and complexity of Presburger (\cite{GradelW2022,Wilke2022}) and real arithmetic (\cite{HanVir2022,HKBV2020}).

\emph{Causal teams}, proposed by Barbero and Sandu \cite{BarSan2020}, fuse together teams and causal models, and model inferences encompassing both functional dependencies arising from data and causal dependencies arising from structural equations.
The logics considered by Barbero and Sandu 
use atomic expressions of the form $X=x$ and $\dep{X}{Y}$ to state that the variable $X$ takes the value $x$ and that (in the data) the value of the variable $Y$ is functionally determined by the values of the variable $X$, respectively. 
\emph{Interventionist counterfactuals} ($X=x\cf \psi$) and \emph{selective implications} ($\alpha \supset \psi$) then describe consequences of actions and consequences of learning from observations. 
For example, the intended reading of the formula
\(
``\textrm{Pressure} = 300 \cf \textrm{Volume}= 4"
\)
is: \emph{If we raise the pressure to $300$ kPa, the volume of the gas will be $4$ $\mathit{m}^3$.}   
On the other hand, the intended reading of the formula
\(
``\textrm{Pressure} = 20 \supset 10<\textrm{Altitude}<30"
\)
is: \emph{If we read $20$ kPa from the barometer, the current altitude is between 10 and 30 km}.    

Finally, the \emph{causal multiteam semantics} coined by Barbero and Sandu \cite{BarSan2024} fuses together multiteams and causal models. The shift from teams to multiteams makes it now possible to study probabilistic conditioning and causal interventions in a unified framework. Barbero and Sandu study a language called $\PCO$ (for Probabilities, Causes and Observations) which they claim to capture a fair portion of the probabilistic causal reasoning that appears in the field of causal inference. It does indeed suffice to capture many forms of probabilistic conditioning, and it suffices to express conditional  $do$ expressions,  the ``Pearl counterfactuals'' mentioned above and more general kinds of statements. For example, the statement ``the probability that a sick untreated patient would be healed when treated is at least $\frac{2}{3}$'' can be formalised as
\(
(\textrm{Sick}=1 \land \textrm{Treated}=0) \supset (\textrm{Treated}=1 \cf \Pr(\textrm{Sick}=0) \geq \frac{2}{3}).
\)
The paper \cite{BarSan2024} raises however the doubt whether $\PCO$ can express, in general, the comparison of conditional probabilities (e.g., statements of the form $\Pr(\alpha\mid\beta)\geq \Pr(\gamma\mid\delta)$). We show here that it fails to do so; thus, $\PCO$ cannot be used, for instance, to compare the expected efficacy of two distinct (non-enforced) medical treatments. The proof shows that this does not improve if we add to $\PCO$ the so-called \emph{strict tensor disjunction} that is often used in the literature on multiteam semantics.

The cornerstone of this inexpressibility result is an abstract characterization of the expressive power of $\PCO$, which in particular shows that the classes of probability distributions that are consistent with a given $\PCO$ formula can be described in terms of a certain class of linear inequalities. On the other hand, by a geometrical argument we see that there are statements of comparison of conditional probabilities which unavoidably involve inequalities of second degree. The quest for an understanding of language $\PCO$ naturally proceeds via an understanding of the expressivity of its key resources:  probabilistic evaluation atoms ($\Pr(\alpha) \geq \epsilon$), comparison atoms ($\Pr(\alpha)\geq \Pr(\beta)$), observations ($\supset$) and interventions ($\cf$). This leads us to the study of four fragments $\PP^-$, $\PP$, $\PO$, and $\PC$. We characterize the expressive power of each of these sublogics, as well as the expressivity of $\PCO$, in terms of closure properties and of an appropriate class of linear inequalities. We also consider how adding the strict tensor operator $\lor$ to these languages may increase the expressive power; we see that $\lor$ behaves as a convex hull operator in the geometry of the probability sets of formulas, and use this fact to show that such extensions are associated to definability by arbitrary linear inequalities. These results are schematized in Table \ref{table:exp}. Together with  geometric reasoning, these characterizations yield a strict hierarchy of expressive power, as summarized in Figure \ref{fig:FIRSTCAUSALGRAPH}. The table and the figure also include a language $\PCO^\omega$ that extends $\PCO$ with (countably) infinite disjunctions. The paper \cite{BarSan2024} already shows that this language is more expressive than $\PCO$; our results yield an alternative proof.

The characterization and hierarchy results for $\PCO$ and its fragments can be found in Section \ref{sec: expressive power}, after a presentation of the semantics and syntax of the languages (Section \ref{sec: preliminaries}). Section \ref{sec: tensor} extends the classification to languages featuring the strict tensor. Section \ref{sec:teamsemantics} presents the inexpressibility results for conditional comparison atoms, and briefly discusses the related issue of definability of dependencies and independencies.

\begin{table}[t]
\begin{center}
\small
\begin{tabular}{ccccc}\toprule
Logic &  & \multicolumn{2}{c}{Closure properties} & References \\ \cmidrule{3-4}
& Type of  & change of & rescaling \& & \\
& inequalities & laws & empty multiteam  & \\
\midrule
$\PP^-$ & monic & X & X & Thm. \ref{thm: characterization of P minus}\\
$\PP$ &  \begin{tabular}{@{}c@{}}monic +  \\ homogeneous signed monic \end{tabular}
& X & X & Thm. \ref{thm: characterization of P}\\
%
$\PO$ & \begin{tabular}{@{}c@{}}monic +  \\ homogeneous signed binary \end{tabular} & X & X  & Thm. \ref{thm: characterization of PO}\\
$\Psm,\PPs$ & linear & X & X  & Thm. \ref{thm: characterization of Ps}\\
$\PC$ & union of signed monic sets &  & X  & Thm. \ref{thm: characterization of PC} \\
$\PCO$ & union of signed binary sets&  &  X  & Thm. \ref{thm: characterization of PCO}\\
$\PCOs$ & union of linear sets &  &  X  & Thm. \ref{thm: characterization of PCOs}\\
$\PCO^\omega$ & (unrestricted)  &  & X &  \cite{BarSan2024}
\\\bottomrule\\
\end{tabular}
\caption{Characterizations of expressivity of logics.
E.g., a class $\mathcal{K}$ of causal multiteams is definable by a $\PO$-formula iff $\mathcal{K}$ is signed binary (see Definition \ref{def:signed_binary}), closed under change of laws and rescaling, and has the empty multiteam property. \corrf{We write e.g. that
$\K$ is a union of signed binary sets to indicate that $\K = \bigcup_{\F\in\mathbb{F}_\sigma}\K^\F$, where $\K^\F$ is a class of signed binary causal multiteams of function component $\F$.}} 
\label{table:exp}
 \end{center}
 \end{table}

  

		
		 


		

		
		
		
		
		

\vspace{-30pt}

\begin{figure}[t]
	\centering
		\begin{tikzpicture}

\normalsize

		\node(po) at (0,1.4) {$\PO$};

         \node(pm) at (-4.0,0) {$\PP^-$};	
		
		\node(p) at (-2.0,0) {$\PP$};
		 
		\node(pco)  at (2, 0) {$\PCO$};

        \node(pc) at (0,-1.4) {$\PC$};

        \node(inf) at (7,0) {$\PCO^\omega$};

        \node(ps) at (3,1.4) {$\Ps$};

        \node(pcos) at (4.5,0) {$\PCOs$};
		
		\draw [->]  (pm)--(p) node[above, pos=0.5]{\footnotesize Lemma \ref{lemma: not all signed monic are monic}};

		\draw [->] (p)--(pc) node[above, sloped, pos=0.5]{\footnotesize Cor. \ref{cor: easy comparisons}};		
		
		\draw[->] (p)--(po) node[above, sloped, pos=0.5]{\footnotesize Prop. \ref{prop: PO greater than P}}; 
		
		\draw[->] (po)--(pco) node[above, sloped, pos=0.5]{\footnotesize Cor.\ref{cor: easy comparisons}};
		
		\draw [->] (pc)--(pco) node[above, sloped, pos=0.5]{\footnotesize Prop. \ref{prop: PO greater than P}}; 

        \draw[->] (po)--(ps) node[above, sloped, pos=0.5]{\footnotesize Cor. \ref{cor: Ps separation}};

        \draw[->] (ps)--(pcos) node[above, sloped, pos=0.5]{\footnotesize Cor. \ref{cor: easy comparisons}};  

        \draw[->] (pco)--(pcos) node[above, pos=0.5]{\footnotesize Prop. \ref{prop: PCO strictly included in PCOs}};
  
		\draw[->] (pcos)--(inf) node[above, pos=0.5]{\footnotesize Prop. \ref{prop: Ps strictly included in PCOs}};

		\end{tikzpicture}
		\caption{Arrows denote strict inclusion of expressivity; $\PC$ and $\PO$ are incomparable (Prop. \ref{prop: PO greater than P} and Cor. \ref{cor: easy comparisons}), as are $\PCO$ and $\Ps$ , and $\PC$ and $\Ps$ (Prop. \ref{prop: PCO Ps incomparable}).
  }
  \label{fig:FIRSTCAUSALGRAPH}
  \vspace{-5mm}
\end{figure}

\vspace{20pt}

\section{Logics with causal multiteam semantics}\label{sec: preliminaries}
Capital letters such as $X, Y,\dots$ denote \textbf{variables} (standing
for specific magnitudes such as ``temperature'' and ``volume'') which
take \textbf{values} denoted by small letters. The values of the variable $X$ will be often
denoted by $x, x', \dots$. 
Sets (and tuples, depending on the context) of variables and values are denoted by boldface letters such as $\SET X$ and $\SET x$.
We consider probabilities that arise from the counting measures of finite (multi)sets. For finite sets $S\subseteq T$, we define 
\(
P_T(S):= \frac{|S|}{|T|}.
\)

A \textbf{signature} is a pair $(\dom,\ran)$, where $\dom$ is a finite set of variables and $\ran$ a function mapping each $X\in \dom$ to a finite set $\ran(X)$ of values (the \textbf{range} of $X$). 
We stipulate a fixed ordering on $\dom$, and write $\SET W$ for the tuple of all the variables of $\dom$ listed in that order. We write $\SET W_X$ for the variables of $\dom\setminus\{X\}$ listed according to the fixed order. 
For a tuple $\SET X = (X_1,\dots, X_n)$ of variables, $\ran(\SET X)$ denotes the Cartesian product $\ran(X_1)\times\dots\times \ran(X_n)$. 
An \textbf{assignment} of signature $\sigma$ is a mapping $s:\dom\rightarrow\bigcup_{X\in \dom}\ran(X)$ such that $s(X)\in \ran(X)$
for each $X\in \dom$.
The set of all assignments of signature $\sigma$ is denoted by $\B_\sigma$.
For an assignment $s$ having the variables of $\SET X = (X_1, \dots, X_n)$ in its domain, $s(\SET X)$ denotes the tuple $(s(X_1), \dots,s(X_n))$.
For $\SET X\subseteq \dom$, $s_{\upharpoonright \SET X}$ is the restriction of $s$ to the variables in $\SET X$.

A \textbf{team} $T$ of signature $\sigma$ is a subset of $\B_\sigma$.
Intuitively, a multiteam is just a multiset analogue of a team.
%
%
%
We represent \textbf{multiteams} as (finite) teams with an extra variable $Key$ (not belonging to the signature) ranging over $\mathbb N$, which takes different values over different assignments of the team, and which is never mentioned in the formal languages. \corrf{We will refer to assignments with the extra variable $Key$ as \textbf{extended assignments}.}  
%
%
A multiteam can be then presented as a table; e.g., the following
\begin{center}
$T$: \begin{tabular}{|c|c|c|}
\hline
 \multicolumn{3}{|l|}{\hspace{-4pt} Key \ \  X \ \    Y} \\
\hline
 \phantom{a}0\phantom{a} & 0 & 0 \\
\hline
 1 & 0 & 0 \\
\hline
 2 & 0 & 1 \\ 
\hline
\end{tabular}
\end{center}
describes a multiteam containing two ``copies'' of the assignment $s(X,Y) = (0,0)$ (first two rows) plus another assignment $t(X,Y) = (0,1)$. We will say that the variable domain of this multiteam $T$ is $\dom=\{X,Y\}$, and omit mentioning the $Key$ variable. Multiteams will be used to encode probability distributions over the underlying team (in this case, the distribution that assigns probability $\frac{2}{3}$ to assignment $s$, and probability $\frac{1}{3}$ to $t$\corrf{; note that we are using the counting probability distribution over a set of \emph{extended} assignments}). The ``underlying team'' (i.e., support of a multiteam) is characterized formally later in Definition \ref{def:support}.

Multiteams by themselves do not encode any solid notion of causation; they do not tell us how a system would be affected by an intervention. We therefore need to enrich multiteams with additional structure. In particular, we will associate to some of the variables a deterministic causal law. The law for variable $V$ takes the form of a function, which describes the way the value of $V$ is generated from the values of other variables in the system. These laws will be used crucially in order to compute how the model is affected by an intervention. Furthermore, we will require that each assignment in the multiteam agrees with these laws.

\begin{definition}\label{def: causal multiteam}
A \textbf{causal multiteam} of signature $(\dom, \ran)$ 
with \textbf{endogenous variables} $\en(T)\subseteq \dom$ is a pair $T = (T^-,\F)$ such that:
\begin{enumerate}
\item 
$T^-$ is a multiteam of domain $\dom$,
\item
$\F$ is a function $\{(V,\F_V) \ | \ V\in \en(T)\}$ that assigns to each endogenous variable $V$ a non-constant $|\SET W_V|$-ary function $\F_V: \ran(\SET W_V)\rightarrow 
\ran(V)$,
\item $(T^-,\F)$ satisfies the \textbf{compatibility constraint}: $\F_V(s(\SET W_V))=s(V)$, for all $s\in T^-$ and $V\in \en(T)$.
\end{enumerate}
%
$T^-$ and $\F$ will be called, respectively, the \textbf{multiteam component} and the \textbf{function component} of $T$. 
We write $(\dom(T),\ran(T))$ to denote the signature of the causal multiteam $T$.\looseness=-1
\end{definition}
%
%
%
%
Notice that, due to the compatibility constraint, not all instances for $\en(T)$ and $T^-$ give rise to causal multiteams.
The function component $\mathcal F$ induces a system of structural equations; an equation
\(
V := \mathcal F_V(\SET W_V)
\)
for each variable $V\in \en(T)$. Note that some of the variables in $\SET W_V$ may not  be necessary for evaluating $V$. For example, if $V$ is given by the structural equation $V:= X+1$, all the variables in $\SET W_V\setminus \{X\}$ are irrelevant (we call them \textbf{dummy arguments} of $\F_V$). The set of non-dummy arguments of $\F_V$ is denoted as $\PA_V$ (the set of \textbf{parents} of $V$). 

We associate to each causal multiteam $T$ a \textbf{causal graph} $G_T$, whose vertices are the variables in $\dom$ and where an arrow is drawn from each variable in $\PA_V$ to $V$, whenever $V\in \en(T)$ (see Example \ref{ex:intervention} and picture \ref{fig:intervention} for a depiction). 
The variables in $\dom(T)\setminus \en(T)$ are called \textbf{exogenous} (written $\exo(T)$). \corrf{The exogenous variables are then those for which our model does not provide a causal explanation; our modeling efforts need to stop somewhere, if we want to avoid infinite regress. Sometimes, in the literature on causation, the exogenous variables are also thought of as factors outside control, which cannot be observed or intervened upon; we will not make such assumption in this paper.}

In the present paper we restrict attention to systems of variables that are connected by causal laws that do not form cycles (e.g., we exclude the possibility that $X$ causally affects $Y$, $Y$ causally affects $Z$, and in turn $Z$ affects $X$); such systems are usually called \emph{recursive}. Concretely, we enforce the following convention:

\begin{center}
\emph{Throughout the paper we will implicitly assume that causal multiteams have an acyclic causal graph.}
\end{center}

\noindent While the study of cyclic systems is far from absent from the literature (e.g. \cite{StrWol1960},\cite{Spi1995},\cite{Hal2000},\cite{BarGal2022}), in a probabilistic context it introduces a number of complications that go well beyond the scope of the framework considered in this paper.


\begin{definition}
 A causal multiteam $S=(S^-,\F_S)$ is a \textbf{causal sub-multiteam} of $T=(T^-,\F_T)$, if they have the same signature, $S^-\subseteq T^-$, and $\mathcal{F}_S = \mathcal{F}_T$. We then write $S\leq T$.
\end{definition}

%

We consider causal multiteams as dynamic models, that can be affected by observations and interventions.
Given a causal multiteam $T = (T^-,\F)$ and a formula $\alpha$ of some formal language (evaluated over 
causal multiteams according to some semantic relation $\models$), ``observing $\alpha$'' produces the causal sub-multiteam $T^\alpha = ((T^\alpha)^-,\F)$ of $T$, where
\(
(T^\alpha)^- \dfn \{s\in T^-  \ | \  (\{s\},\F)\models \alpha\}.
\)\footnote{Throughout the paper, the semantic relation in terms of which $T^\alpha$ is defined will be the semantic relation for language $\CO$, which shall be defined below.}
On the other hand, an intervention on $T$ will \emph{not}, in general, produce a sub-multiteam of $T$. It will instead modify the values that appear in some of the columns of $T$. We consider interventions that are described by conjunctions of the form $X_1=x_1 \land\dots\land X_n=x_n$ (or, shortly, $\SET X = \SET x$). Such a formula is \textbf{inconsistent} if there are two indexes $i,j$ such that $X_i$ and $X_j$ denote the same variable, while $x_i$ and $x_j$ denote distinct values; it is \textbf{consistent} otherwise.
Applying an intervention $do(\SET X = \SET x)$, where $\SET X = \SET x$ is consistent, to a causal multiteam $T = (T^-,\F)$ of endogenous variables $\SET V$ will produce a causal multiteam $T_{\SET X = \SET x} = (T^-_{\SET X = \SET x},\F_{\SET X = \SET x})$ of endogenous variables $\SET V \setminus \SET X$, with:
\begin{itemize}
\item function component  
\(
\F_{\SET X = \SET x} := \F_{\upharpoonright(\SET V \setminus \SET X)}
\)
(the restriction of $\F$ to the set of variables $\SET V \setminus \SET X$) 

\item multiteam component  
\(
T_{\SET X=\SET x}^-:=\{s^\F_{\SET X=\SET x}\mid s\in T^-\},
\)
where each $s^\F_{\SET X=\SET x}$ is the unique assignment compatible with $\mathcal F_{{\SET X = \SET x}}$ defined (recursively) as 
\[s^\F_{\SET X=\SET x}(V)=\begin{cases}
x_i&\text{ if }V=X_i\in \SET X\\
s(V)&\text{ if }V\in \exo(T)\setminus \SET X\\
\F_V(s^\F_{\SET X=\SET x}(\SET W_V
))&\text{ if }V\in \en(T)\setminus \SET X.
\end{cases}
\]
\end{itemize}
We emphasize that the uniqueness of $s^\F_{\SET X=\SET x}$, and thus the correctness of this definition, hinges on our assumption that the causal graphs are acyclic. For an explanation of how interventions may be defined in the cyclic (non-probabilistic) case, see \cite{BarGal2022}.

\begin{figure}[t]
\begin{center}
$T^-$: \begin{tabular}{|c|c|c|c|}
\hline
\multicolumn{4}{|c|}{ } \\
\multicolumn{4}{|l|}{\small{Key} \ \ \ \  $X$\tikzmark{XR} \ \ \ \ \tikzmark{YL}$Y$\tikzmark{YR} \ \ \ \ \tikzmark{ZL}$Z$} \\
\hline
 $\phantom{a}$0$\phantom{a}$ & 0 & 1 & 1\\
\hline
 1 & 1 & 2 & 3 \\
\hline
 2 & 1 & 2 & 3 \\
\hline
 3 & 2 & 3 & 5\\ 
\hline
 4 & 2 & 3 & 5\\ 
\hline
 5 & 2 & 3 & 5\\ 
\hline
\end{tabular}
\begin{tikzpicture}[overlay, remember picture, yshift=0\baselineskip, shorten >=.5pt, shorten <=.5pt]
  \draw ([yshift=7pt]{pic cs:XR})  edge[line width=0.2mm, out=45,in=135,->] ([yshift=7pt]{pic cs:ZL});
  \draw [->] ([yshift=3pt]{pic cs:YR})  [line width=0.2mm] to ([yshift=3pt]{pic cs:ZL});
  \draw [->] ([yshift=3pt]{pic cs:XR})  [line width=0.2mm] to ([yshift=3pt]{pic cs:YL});
\end{tikzpicture}
\hspace{3mm}
$\leadsto$ \hspace{5pt} \begin{tabular}{|c|c|c|c|}
\hline
\multicolumn{4}{|c|}{ } \\
\multicolumn{4}{|l|}{  \small{Key} \ \ \ \  $X$\tikzmark{XR'}  \ \ \ \   \tikzmark{YL'}$Y$\tikzmark{YR'} \ \ \ \ \tikzmark{ZL'}$Z$} \\
\hline
 $\phantom{a}$0$\phantom{a}$ & 0 & \textbf{1} & ... \\
\hline
 1 & 1 & \textbf{1} &  ...  \\
\hline
 2 & 1 & \textbf{1} &  ... \\
\hline
 3 & 2 & \textbf{1} &  ... \\ 
\hline
 4 & 2 & \textbf{1} & ... \\ 
\hline
 5 & 2 & \textbf{1} &  ... \\ 
\hline
\end{tabular}
\hspace{10pt} $\leadsto$ 

\vspace{10pt}

\hspace{-60pt} $\leadsto$ \hspace{5pt}  $T_{Y=1}^-$: \begin{tabular}{|c|c|c|c|}
\hline
\multicolumn{4}{|c|}{ } \\
\multicolumn{4}{|l|}{\small{Key} \ \ \ \ $X$\tikzmark{XR''} \ \ \ \ \tikzmark{YL''}$Y$\tikzmark{YR''} \ \ \ \ \tikzmark{ZL''}$Z$} \\
\hline
 \phantom{a}0\phantom{a} & 0 & 1 & \textbf{1}\\
\hline
 1 & 1 & 1 & \textbf{2} \\
\hline
 2 & 1 & 1 & \textbf{2} \\
\hline
 3 & 2 & 1 & \textbf{3}\\ 
\hline
 4 & 2 & 1 & \textbf{3}\\ 
\hline
 5 & 2 & 1 & \textbf{3}\\ 
\hline
\end{tabular}
\begin{tikzpicture}[overlay, remember picture, yshift=0\baselineskip, shorten >=.5pt, shorten <=.5pt]
  \draw ([yshift=7pt]{pic cs:XR'})  edge[line width=0.2mm, out=45,in=135,->] ([yshift=7pt]{pic cs:ZL'});
  \draw [->] ([yshift=3pt]{pic cs:YR'})  [line width=0.2mm] to ([yshift=3pt]{pic cs:ZL'});
  
  \draw ([yshift=7pt]{pic cs:XR''})  edge[line width=0.2mm, out=45,in=135,->] ([yshift=7pt]{pic cs:ZL''});
  \draw [->] ([yshift=3pt]{pic cs:YR''})  [line width=0.2mm] to ([yshift=3pt]{pic cs:ZL''});
\end{tikzpicture}
\end{center}
\caption{Causal multiteams for Example \ref{ex:intervention}, showing how the multiteam component $T_{Y=1}^-$ of a causal multiteam is computed from $T^-$ given an  intervention $do(Y=1)$. The figure also  describes the associated causal graphs.}
\label{fig:intervention}

\end{figure}

\begin{example}\label{ex:intervention}
Consider the causal multiteam $T= (T^-, \F)$ depicted in Figure \ref{fig:intervention},
where each row of the leftmost table  depicts an assignment of $T^-$ (e.g., the third row represents an assignment $s$ with $s(Key)=2$, $s(X)=1$, $s(Y)=2$, $s(Z)=3$).  The rows of the table are compatible with the laws $\F_Z(X,Y) = X+Y$ and $\F_Y(X)=X+1$, while $X$ is exogenous. $T$ encodes 
probabilities for formulas that discuss variables $X,Y,Z$ and their possible values; for example, $P_T(Z=3) = \frac{1}{3}$.

Suppose we can enforce the variable $Y$ to take the value $1$. The effect of such an intervention, depicted in the right-hand side of Figure \ref{fig:intervention}, is to first set the value of $Y$ to $1$ (in all rows) and then to recompute the values of $Z$ using the function $\F_Z$. The probability distribution has changed: now  $P_{T_{Y=1}}(Z=3)=\frac{1}{2}$. Furthermore, the function $\F_Y$ is omitted from $T_{Y=1}$, and thus the arrow from $X$ to $Y$ has been omitted from the causal graph.   
\end{example}

Given two languages $\La,\La'$ of signature $\sigma$, whose semantics is defined over causal multiteams, and formulae $\varphi\in \La$ and $\varphi'\in \La'$, we write $\varphi \equiv_\sigma \varphi'$ if $T\models \varphi \Leftrightarrow T\models \varphi'$ holds for all causal multiteams $T$ of signature $\sigma$. We omit the index $\sigma$ if it is clear from the context.
Similarly, we may write $\mathcal{L}_\sigma$ to emphasise that the signature of $\La$ is $\sigma$.

We write $\La\leq\La'$
if for every $\varphi\in\La$ there is $\varphi'\in\La'$ with $\varphi \equiv \varphi'$.
We write $\La<\La'$ 
if $\La\leq\La'$ but $\La'\not\leq\La$.
Finally, we write $\La \equiv \La'$
if $\La\leq\La'$ and $\La'\leq\La$. 
$\K^\sigma_\varphi$ is  the set 
of all causal multiteams of signature $\sigma$ that satisfy $\varphi$. $\K^\sigma_\varphi$ will be (with the exception of contradictory formulae) a countably infinite set.

A class $\K$ of causal multiteams is \textbf{definable} in $\mathcal L_\sigma$ if $\K=\K^\sigma_\varphi$ for some $\varphi\in\La_\sigma$.

A class $\K$ is \textbf{flat} if $(T^-,\F)\in\K$ iff $(\{s\},\F)\in \K$ for every $s\in T^-$. A class $\K$ of causal multiteams of signature $\sigma$ has the \textbf{empty multiteam property}, if $\K$ includes all empty causal multiteams of signature $\sigma$ (we say that a causal multiteam $(T^-,\F)$ is \textbf{empty} if the multiteam $T^-$ is). A $\sigma$-formula $\varphi$ has one of the above (or to be defined) properties, if $\K^\sigma_\varphi$ has it. A language $\La$ is flat ( resp. has the empty team property), if every $\varphi\in \La$ is flat (resp. has the empty team property). In general, we say that $\La$ has a certain property if and only if each $\varphi\in \La$ has it.

The language $\CO$, introduced in \cite{BarSan2020}, is defined by the following BNF grammar:
\[
\alpha ::= Y=y \mid Y\neq y \mid  \alpha\land\alpha  \mid  \alpha\lor\alpha  \mid  \alpha\supset\alpha  \mid  \SET X = \SET x \cf \alpha,
\]
where $\SET{X}\cup\{Y\}\subseteq \dom$, $y\in \ran(Y)$, and $\SET x\in \ran(\SET X)$.
It is a language for the description of facts. We will later introduce extensions that allow us to talk about the probabilities of the facts that are expressible in $\CO$.
Formulae of the forms $Y = y$ and $Y\neq y$ are \textbf{literals}.
The semantics for $\CO$ is given by the following clauses:
\begin{align*}
&T\models Y=y &&\text{iff}&& s(Y)=y \text{ for all } s\in T^-.\\
&T\models Y\neq y  &&\text{iff}&& s(Y)\neq y \text{ for all }s\in T^-.\\
&T\models \alpha\land \beta  &&\text{iff}&& T\models \alpha \text{ and } T\models \beta.\\
&T\models \alpha\lor \beta &&\text{iff}&& \text{there are $T_1,T_2\leq T$ s.t. } T_1^-\cup T_2^- = T^-,\\ & && && T_1^-\cap T_2^- = \emptyset,  T_1\models \alpha \text{ and  }T_2\models \beta.\\
&T\models \alpha \supset \beta &&\text{iff} && T^\alpha \models \beta.\\
&T\models \SET X=\SET x \cf \beta &&\text{iff}&& T_{\SET X=\SET x} \models \beta \text{ or } \SET X=\SET x \text{ is inconsistent}.
\end{align*}
where $T^\alpha$ is defined simultaneously with the clauses, as previously explained.

The intuitive readings of the conditional formulas $\alpha \supset \beta$ and $\SET X=\SET x \cf \beta$ are, respectively, ``After observing (or learning) $\alpha$, certainly $\beta$ holds'' and ``After setting $\SET X$ to $\SET x$, certainly $\beta$ holds''. Some of the semantic clauses for the other connectives may look unusual to a  reader unaccustomed to team semantics, but they are natural lifts of the usual Tarskian clauses from a setting in which formulas are evaluated on single assignments to a setting where they are evaluated on a multiplicity of assignments (for an overview of team semantics, the reader may consult e.g. \cite{DurKonVol2016}). As an example, the clause for a disjunction $\alpha\lor\beta$ is just stating that each assignment in $T$ satisfies either $\alpha$ or $\beta$. It says so by saying that $T$ can be split into two parts, one containing assignments that satisfy $\alpha$ and one containing assignments that satisfy $\beta$. This reading of the clauses is made possible by the fact that language $\CO$ is flat. The proof of the following result is similar to that of the analogous result for causal teams \cite[Thm. 2.10]{BarSan2020}.
\begin{theorem}\label{thm: CO flatness}
$\CO_{\sigma}$ is flat and therefore has the empty multiteam property. 
\end{theorem} 
In a sense, flatness tells us that $\CO$ behaves as a classical language. The probabilistic languages that we shall consider later will not be flat; probabilistic statements are meaningful at the level of multiteams but not at the level of the single assignments.

We also remark that in \cite{BarSan2020} the operator $\lor$ was defined without insisting  that $T_1^-\cap T_2^- = \emptyset$.
This was done since the paper considered set-based semantics.
As our semantics is based on multisets, the appropriate definition of $\lor$ uses a union that is sensitive to multiplicities (i.e. disjoint union).
Theorem \ref{thm: CO flatness} 
entails 
that this distinction is irrelevant for $\CO$, but it will have an impact when we consider languages in which $\lor$ can be applied to probabilistic formulas. Then the two definitions disagree, because such languages lack the 
property called \emph{downward closure}: if $T\models \varphi$ and $S\leq T$,  then $S\models\varphi$.






If we pick a variable $X$ in the signature and a value $x\in Ran(X)$, we can abbreviate the formulae $X=x \lor X\neq x$ and $X=x \land X\neq x$ as $\top$, resp. $\bot$ (the former is a valid formula because it just says that the multiteam can be split in two parts, the assignments where $X$ takes value $x$ and those where it does not).
The so-called \textbf{dual negation} of a formula $\alpha$,
\(
T\models \neg\alpha \text{ iff } (\{s\},\F)\not\models  \alpha \text{ for all }s\in T^-,
\)
can be taken in $\CO$ as an abbreviation for $\alpha\supset \bot$. We remark that $\neg \alpha$ also admits an (inductive) definition that does \emph{not} use the symbol $\supset$, as follows:
\begin{itemize}
\item $(X = x)^d$ is $X \neq x$ 
\item $(X \neq x)^d$ is $X = x$ 
\item $(\alpha \land \beta)^d$ is $(\alpha^d \lor\beta^d)$
\item $(\alpha \lor \beta)^d$ is $(\alpha^d \land \beta^d)$
\item $(\alpha\supset\beta)^d$ is $\alpha \land \beta^d$
\item $(\SET X = \SET x\cf\beta)^d$ is $\SET X = \SET x\cf\beta^d$
\end{itemize} 
One can then show that $\neg\alpha \equiv \alpha^d$, analogously as in \cite{BarSan2020}, Theorem 2.11.


 

Next, we introduce 
a language with \textbf{probabilistic atoms}
$\Pr(\alpha) \geq \epsilon$, $\Pr(\alpha) > \epsilon$,  $\Pr(\alpha) \geq \Pr(\beta)$, $\Pr(\alpha) > \Pr(\beta)$,
where $\alpha,\beta\in\CO$ and $\epsilon \in [0,1]\cap \mathbb{Q}$. The first two are called \textbf{evaluation atoms}, and the latter two \textbf{comparison atoms}. 
Probabilistic atoms 
 together with literals of $\CO$ are called \textbf{atomic formulae}.
The probabilistic language $\PCO$ is then given by the following grammar:
\[
\varphi::= \eta \mid \varphi \land \varphi \mid  \varphi \sqcup \varphi \mid \alpha\supset \varphi \mid \SET X = \SET x\cf \varphi,
\]
where $\SET X\subseteq \dom$, $\SET x \in \ran(\SET X)$, $\eta$ is an atomic formula, 
and $\alpha$ is a $\CO$ formula.
Note that the antecedents of $\supset$ and the arguments of probability operators are $\CO$ formulae.
%
The semantic clauses for the additional operators are given below:
\begin{align*}
&T\models \psi\sqcup \chi &&\text{iff}&& T\models \psi \text{ or } T\models \chi\\
&T\models \Pr(\alpha)\geq \epsilon  &&\text{iff}&& T^-=\emptyset \text{ or } P_T(\alpha)\geq \epsilon\\
&T\models \Pr(\alpha) > \epsilon  &&\text{iff}&& T^-=\emptyset \text{ or } P_T(\alpha) > \epsilon\\
&T\models \Pr(\alpha)\geq \Pr(\beta) &&\text{iff}&& T^-=\emptyset \text{ or } P_T(\alpha)\geq P_T(\beta)\\
&T\models \Pr(\alpha) > \Pr(\beta)  &&\text{iff}&& T^-=\emptyset \text{ or } P_T(\alpha) >  P_T(\beta),
\end{align*}
where $P_T(\alpha)$ is a shorthand for $P_{T^-}((T^\alpha)^-)$.\footnote{We remark that in $\PCO$ (but not in $\CO$!) is it also possible to define, inductively, an operator that behaves as classical negation on nonempty causal multiteams (\emph{weak contradictory negation}). Details can be found in \cite{BarVir2023b}; we will not use it here.} 
The language $\PCO$ still has the empty team property 
but it is not flat.
%
The definability of the dual negation in $\CO$ allows us to introduce many useful abbreviations:
\begin{multicols}{2}
\allowdisplaybreaks
\noindent
\begin{align*}
\Pr(\alpha) \leq  \epsilon &\dfn\Pr(\neg\alpha) \geq 1 - \epsilon \\ 
\Pr(\alpha) < \epsilon &\dfn \Pr(\neg\alpha) > 1 - \epsilon\\
\Pr(\alpha) = \epsilon &\dfn \Pr(\alpha) \geq  \epsilon \land\Pr(\alpha) \leq  \epsilon\\
\Pr(\alpha) \neq  \epsilon &\dfn \Pr(\alpha) >  \epsilon \sqcup  \Pr(\alpha) <  \epsilon
\end{align*}
\end{multicols}
\noindent We will see in Section \ref{sec:teamsemantics} that the $\supset$ operator
enables us to express some statements involving conditional probabilities. 

\begin{example}
Let $T = (T^-,\F)$ be a causal multiteam over variables $\mathrm{GroundSpeed}$, $\mathrm{DescentAngle}$, $\mathrm{StructuralIntegrity}$, $\mathrm{SafeLanding}$ depicting data related to landing an Airbus A350-900 aircraft.
The first three variables are numerical, while the last is Boolean.
The structural equation $F_\mathrm{SL}(\mathrm{GS},\mathrm{DA}, \mathrm{SI})$ outputs a Boolean value ``true'' when a plane of given structural integrity is expected to make a safe landing at a given speed and angle. 
The formula
\(
``\mathrm{SI} \neq 0 \supset [(\mathrm{GS}=300 \land \mathrm{DA}=4)\cf \Pr(\mathrm{SL} = \mathrm{false})<0.01]"
\)
expresses that the probability of landing failure is less than $1\%$ when setting a landing speed of $300km/h$ and descent angle of 4 degrees, conditional on the plane not being grounded due to structural condition ($\mathrm{SI}= 0 $).

Since we can assume that $\mathrm{SI}$ is exogenous (the assessment of structural integrity is not affected by the speed and angle set during the flight), this statement can be equivalently written as
\(
`` (\mathrm{GS}=300 \land \mathrm{DA}=4)\cf (\mathrm{SI} \neq 0 \supset\Pr(\mathrm{SL} = \mathrm{false})<0.01)".
\)
This would not be legitimate if $\mathrm{SI}$ was causally affected by $\mathrm{GS}$ or $\mathrm{DA}$; 
the operators $\cf$ and $\supset$ do not in general commute with each other. 

\end{example}




We consider the following syntactic fragments of $\PCO$, which preserve the syntactic restrictions yielded by its two level syntax -- that the antecedents of $\supset$ and the arguments of $\Pr$ are always $\CO$ formulae. $\PP$ is the fragment without $\supset$ and $\cf$.
$\PP^-$ is the fragment of $\PP$ without comparison atoms. $\PC$ and $\PO$ are fragments of $\PCO$ without $\supset$ and $\cf$, respectively.

Finally, we will consider two extensions of $\PCO$. $\PCO^\omega$ is the extension of $\PCO$ with countable disjunctions of the form
\(
\bigsqcup_{i\in I} \psi_i,
\)
where the $\psi_i$ are $\PCO$ formulae.

The language $\PCOs$ is the extension of $\PCO$ where the operator $\lor$ can be arbitrarily used (except in antecedents of $\cf$). More precisely, its grammar is
\[
\varphi::= \eta \mid \varphi \land \varphi \mid  \varphi \sqcup \varphi \mid \varphi \lor \varphi \mid \alpha\supset \varphi \mid \SET X = \SET x\cf \varphi,
\]
where, as before, $\eta$ is an atom and $\alpha$ is a $\CO$ formula. We will also denote as $\Ps$ the fragment of $\PCOs$ without $\supset,\cf$, and as $\Psm$ the fragment of $\Ps$ without comparison atoms.

%

\section{Expressive power of fragments of $\PCO$}\label{sec: expressive power}

We start by rephrasing the known characterizations from the literature. A number of results appear in the literature (e.g. in \cite{BarYan2022}) that characterize causal languages in the context of causal \emph{team} semantics. A causal team (of signature $\sigma$) is, essentially, a pair $(T^-,\F)$, where $T^-$ is a team instead of a multiteam (i.e., a set of assignments on $Dom$ instead of $Dom \cup \{Key\}$), satisfying the conditions given in Definition \ref{def: causal multiteam}. Each causal multiteam can be seen as a causal team enriched with a probability distribution. This correspondence is expressed precisely as follows:
\begin{definition}\label{def:support}
  The \textbf{support} of a causal multiteam $T=(T^-,\F)$ is the causal team $\Team(T) = (\Team(T^-),\F)$, where $\Team(T^-):=  \{s_{\upharpoonright Dom} \mid s\in T^-\}$.  
\end{definition}
It is immediate to see that a language without probabilistic features (such as $\CO$) cannot tell apart two causal multiteams that have the same support. From this, it is straightforward but tedious (the details can be found in the Appendix) 
to show that the characterization of $\CO$ given in \cite[Theorem 4.4]{BarYan2022} in terms of causal teams holds unchanged over causal multiteams:


\begin{restatable}[Characterization of $\CO$]{theorem}{COchar}\label{thm:COchar}
Let $\sigma$ be a finite signature, and $\K$ a class of causal multiteams of signature $\sigma$. Then $\K$ is definable by a $\CO_{\sigma}$ formula (resp. a set of $\CO_{\sigma}$ formulae) if and only if $\K$ is flat.
\end{restatable}

 $\PCO$ is a purely probabilistic language; it cannot tell apart multiteams that represent one and the same distribution. Given an assignment $t$ and a causal team $T = (T^-,\F)$, we write $\#(t,T)$ for the number of copies of $t$ in $T^-$ and (provided $T$ is nonempty) $\epsilon_t^T := \frac{\#(t,T)}{|T^-|}$ for the probability of $t$ in $T$.
Two causal teams $S =   (S^-,\F)$ and $T=(T^-,\G)$ are \textbf{rescalings} of each other ($S\sim T$) if $\F=\G$ and either $S^- = T^- = \emptyset$ or $\epsilon_t^T =\epsilon_t^S$ for each assignment $t$.
A class $\K$ of causal multiteams of signature $\sigma$ is \textbf{closed under rescaling} if, whenever $S\in \K$ and  $S \sim T$, also $T\in \K$.
An ideal language for purely probabilistic reasoning should be characterized just by this condition. It turns out that $\PCO$ is not expressive enough for the task, however its extension with countable global disjunctions $\PCO^\omega$ is (modulo the empty multiteam property).

\begin{theorem}[\cite{BarSan2024}]\label{theorem: expressivity of PCOinf}
A \emph{nonempty} class $\K$ of multiteams of signature $\sigma$ is definable in 
 $\PCO^{\omega}_{\sigma}$ 
 \corrf{(either by a formula or by a set of formulae)}
  iff $\K$ has the empty multiteam property and is closed under rescaling.
\end{theorem}

\noindent The key to the proof is the fact that
for any causal multiteam $(T^-, \F)$ one can write $\PCO$-formulae $\Theta_{T^-}$ and $\Phi^\F$ that characterize the properties of having team component $T^-$ (up to rescaling) and function component $\F$, respectively.
A set $\K$ of causal multiteams is then defined by the formula $\bigsqcup_{(T^-,\F)\in \K }(\Theta_{T^-} \land \Phi^\F)$.
Since $\K$ can be countably infinite,  the proof crucially depends on the use of infinitary disjunctions and gives us no hints on how to obtain a finitary logic with such expressivity. Actually, a counting argument given in \cite{BarSan2024} shows that such a language must be uncountable, and thus that $\PCO < \PCO^\omega$. Our characterization of the expressivity of $\PCO$ will provide an alternative proof for the strict inclusion.  



In order to characterize the expressivity of $\PCO$ and its fragments, we need to introduce some classes of linear inequalities and closure properties of classes of causal multiteams. For the latter, we have already seen closure under rescaling and the empty multiteam property.
A class $\K$ of causal multiteams of signature $\sigma$ is \textbf{closed under change of laws} if, whenever $(T^-,\F)\in \K$ and $\G$ is a system of functions of signature $\sigma$ such that $(T^-,\G)$ satisfies the compatibility constraint (point 3. of definition \ref{def: causal multiteam}), then  $(T^-,\G)\in\K$.  


It is self-evident that the logics without $\cf$ are closed under change of laws, while the logics with  
$\cf$ 
are not. Thus, the following hold.
\begin{lemma}\label{lemma:trivial laws}
$\PP^-$, $\PP$, $\PO$, $\Psm$ and $\Ps$ are closed under change of laws. $\CO$, $\PC$, $\PCO$ and $\PCOs$ are not closed under change of laws.
\end{lemma}
\begin{corollary}\label{cor: easy comparisons}
1) $\PP < \PC$, $\PO < \PCO$ and $\Ps < \PCOs$

2) $\PC\not\leq \PO$ and $\PC\not\leq \Ps$.
\end{corollary}

\subsection{Monic and signed monic probability sets: $\PP^-$, $\PP$, and $\PC$}\label{subs: Expressivity probabilistic}

 
We characterize the expressivity of fragments of $\PCO$ by investigating the families of subsets of $\Q^n$ that are definable in the logics.
 For a given signature $\sigma$, we fix an enumeration $s_1,\dots,s_n$ of the assignments of $\B_\sigma$; every \emph{nonempty} causal multiteam $T$ can then be associated with a \textbf{probability vector} $\pvec_{T}=(\epsilon_{s_1}^T, \dots , \epsilon_{s_n}^T)\in \mathbb Q^n$.
 Similarly, a class $\K$ of causal multiteams of signature $\sigma$ has an associated \textbf{probability set}  $\Pvec_\K = \{\pvec_T \mid T\in\K, T \text{ nonempty }\}$.
 Note that $\pvec_{T}$ and $\Pvec_\K$ are, respectively,  a point and a subset of the \textbf{standard $n-1$-simplex} $\Delta^{n-1}$ (i.e. the set of points of $[0,1]^n\cap \mathbb Q^n$ that satisfy the equation $\epsilon_{s_1} + \dots + \epsilon_{s_n} =1$). 
 To each formula $\varphi$, we can associate a probability set $\Pvec_\varphi:= \Pvec_{\K_\varphi}$.
 Note that if $S,T$ are causal multiteams of the same signature and same function component, such that $\pvec_S = \pvec_T$, then  $S$ is a rescaling of $T$.
Similarly, a class $\K$ of causal multiteams of signature $\sigma$ that is closed under change of laws and rescaling is the largest class of causal multiteams of signature $\sigma$ having probability set $\Pvec_\K$.
 %
 
 A \textbf{linear inequality} is an expression of the form
 \(
a_1\epsilon_1 +\dots+ a_n\epsilon_n \vartriangleright b,
 \) 
where $\vartriangleright \hspace{3pt} \in \{\geq,\leq,>,<\}$, $a_1,\dots,a_n,b\in\mathbb Q$, and $\epsilon_1,\dots \epsilon_n$ are variables (in the usual algebraic sense). A linear inequality will be called \textbf{homogeneous} if its constant coefficient $b$ is $0$. A linear inequality is \textbf{signed monic} if each of the $a_i$ is in $\{0,1,-1\}$. It is \textbf{monic} if each of the $a_i$ is in $\{0,1\}$. \corrf{We extend this terminology to probability sets as follows.}
\begin{definition}
    A probability set $\Pvec$ is \textbf{monic} if it is a finite union of subsets of $\Delta^{n-1}$ defined by finite systems of monic inequalities.

    A probability set is \textbf{signed monic} if it is defined, over $\Delta^{n-1}$, as a finite union of probability sets defined by finite systems of inequalities of the following two types:
       \begin{enumerate}
       \item monic inequalities 
       \item homogeneous signed monic inequalities. 
       \end{enumerate}
\end{definition}
A class $\K$ of causal multiteams of a fixed signature is \textbf{(signed) monic} if $\Pvec_\K$ is a (signed) monic probability set. A \textbf{polytope} is a subset of $\Delta^{n-1}$ that is defined by a single system of linear inequalities which only use the relations $\leq$ and $\geq$.


We will show that being monic and closed under change of laws and rescaling characterizes 
expressibility in $\PP^-$, whereas being signed monic and closed under change of laws and rescaling characterizes 
expressibility in $\PP$. 

\begin{lemma}\label{lemma: preserve polytopes}
1) If $\Pvec$ is a monic (resp. signed monic) 
polytope, then its complement in $\Delta^{n-1}$ 
 is monic (resp. signed monic).
2) If $\Pvec,\Qvec$ are monic (resp. signed monic), then also $\Pvec \cup \Qvec$ is.
3) If $\Pvec,\Qvec$ are monic (resp. signed monic, polytopes), 
then also $\Pvec \cap \Qvec$ is.
\end{lemma}

\begin{proof} 
1) 
Since $\Pvec$ is a polytope, a point is in $\Pvec$ iff it is a solution to a finite system of monic (resp. \corrf{monic and homogeneous} signed monic) inequalities $e_1,\dots,e_m$, where each $e_i$ is either of the form $a^i_1 \epsilon_1 + \dots + a^i_{k(i)} \epsilon_{k(i)} \geq b^i$ or $a^i_1 \epsilon_1 + \dots + a^i_{k(i)} \epsilon_{k(i)} \leq b^i$.
We denote by $\overline e^i$ the inequality obtained from $e^i$ by replacing $\geq$ with $<$ and $\leq$ with $>$. Each $\overline e^i$ defines a monic (resp. \corrf{monic or homogeneous} signed monic) set $\Qvec_i$. Now, a point is \emph{not} in $\Pvec$ iff it satisfies $\overline e^i$ for some $i$, iff it is an element of $\Qvec_1\cup\dots\cup\Qvec_m$.  The latter is monic (resp. \corrf{monic or} signed monic) by construction. 

2) 
This follows straightforwardly from the definitions.

3) First consider two sets $\Pvec$, $\Qvec$ defined by single finite systems $A$, resp. $B$ of inequalities.  Then their intersection is defined by the single system $A\cup B$. In particular, the intersection of two polytopes is a polytope. 


Then, consider two finite unions $\Pvec_1 \cup\dots\cup\Pvec_m$ and $\Qvec_1\cup\dots\cup\Qvec_n$ of probability sets, and assume wlog that each of  the $\Pvec_i$ and $\Qvec_i$ is defined by a single system of monic (resp. \corrf{monic or} signed monic) inequalities.  
 Then their intersection is $\bigcup_{\substack{ i=1\dots m \\ j = 1\dots n }} (\Pvec_i \cap \Qvec_j)$. By the previous observation and 2., this is again a monic (resp. signed monic) set. \qed
 \end{proof}

\begin{lemma}\label{lemma: P minus is monic}
If $\varphi\in\mathcal P^-_{\sigma}$, then $\K_\varphi^\sigma$ is monic.
\end{lemma}

\begin{proof}
We proceed by induction on $\varphi$ and restrict attention to nonempty causal multiteams $T$ of signature $\sigma$. 

The cases for $\land$ and $\sqcup$ follow directly from Lemma \ref{lemma: preserve polytopes} and the induction hypotheses, since $\Pvec_{\psi\land\chi}= \Pvec_\psi \cap \Pvec_\chi$ and $\Pvec_{\psi\sqcup\chi}= \Pvec_\psi \cup \Pvec_\chi$.
The case for $\varphi = \Pr(\alpha)\geq \epsilon$ follows from the following chain of equivalences: $T\in \K_\varphi^\sigma$ iff $\frac{|(T^\alpha)^-|}{|T^-|} \geq \epsilon$ iff $\sum_{s\in \Team((T^\alpha)^-)} \epsilon_s^T \geq \epsilon$, since the last inequality defines a monic polytope.

The cases for $\varphi: X=x$ (resp. $\varphi: X \neq x$) can be seen as a special case of the previous one, since these formulae are equivalent to $\Pr(X=x)\geq 1$ (resp. $\Pr(X \neq x)\geq 1$) and thus they have the same monic probability sets.

In the case for $\varphi = \Pr(\alpha) > \epsilon$, analogous calculations as above show that the probability set of $\varphi$ is the set of solutions of the monic inequality $\sum_{s\in \Team((T^\alpha)^-)} \epsilon_s^T > \epsilon$.
\qed    
\end{proof}

\begin{lemma}\label{lemma: P is signed monic}
If $\varphi\in\mathcal P_{\sigma}$, then $\K_\varphi^\sigma$ is signed monic.
\end{lemma}

\begin{proof}
The proof, by induction on $\varphi$, proceeds as in the previous lemma. We cover the missing cases. Again, we restrict $T$ to nonempty causal multiteams.

Let $\varphi$ be $\Pr(\alpha)\geq \Pr(\beta)$.  Then $T\in \K_\varphi^\sigma$ iff $\frac{|(T^\alpha)^-|}{|T^-|} \geq \frac{|(T^\beta)^-|}{|T^-|}$ 
 iff $\sum_{s\in \Team((T^\alpha)^-)} \epsilon_s^T + \sum_{s\in \Team((T^\beta)^-)} (-1) \cdot \epsilon_s^T \geq 0$. This inequality defines a signed monic polytope.

Let $\varphi$ be $\Pr(\alpha)> \Pr(\beta)$. By the same calculations as above, a causal multiteam $T$ satisfies this formula iff $\Pvec_T$ is characterized by the signed monic inequality $\sum_{s\in \Team((T^\alpha)^-)} \epsilon_s^T + \sum_{s\in \Team((T^\beta)^-)} (-1) \cdot \epsilon_s^T > 0$.
\qed
\end{proof}



We can now state and prove the semantic characterizations of languages $\PP^-$ and $\PP$. 
Remember that there are only finitely many assignments of signature $\sigma$ (say $s_1,\dots,s_n$). We can then describe each such assignment $s_i$ with a formula $\Al_i \dfn \SET W = s_i(\SET W)$, where $\SET W$ lists all the variables in $Dom$.

\begin{restatable}{theorem}{Pchar}\label{thm: characterization of P minus}
A class $\K$ of multiteams of signature $\sigma$ is definable
in $\PP^-$ if and only if:
\begin{enumerate}
\item $\K$ is closed under change of laws
\item $\K$ is closed under rescaling
\item $\K$ is monic.
\end{enumerate} 
\end{restatable}

\begin{proof}

 Left-to right. The fact that $\PP^-$ and $\PP$ have the empty multiteam property and are closed under rescaling follows from Theorem \ref{theorem: expressivity of PCOinf}; closure under change of laws follows from Lemma \ref{lemma:trivial laws}.
The fact that $\PP^-$ is monic is given by Lemma \ref{lemma: P minus is monic}, and that $\PP$ is signed monic by Lemma \ref{lemma: P is signed monic}.


For $\PP^-$, the right-to-left entailment is proved via a direct translation from finite unions of finite systems of monic inequalities into $\PP^-$ formulae. Suppose that $\K$ has all the four listed properties. Since it is monic, $\Pvec_\K$ is the union of the solution sets of a finite number $m$ of finite systems $A_1,\dots,A_m$ of inequalities. 
For each $A_j$, let us index its inequalities by a finite set $I_j$.

For each inequality $e_i \dfn a_1^i\epsilon_1 + \dots + a_n^i\epsilon_n \vartriangleleft b^i$ in $A_j$, if $b^i \in [0,1]$ define
\[
\psi_i \dfn \Pr(\bigvee_{s_k \in \B_\sigma \mid a_k^i = 1} \Al_k) \vartriangleleft b^i. 
\]
If $b<0$ and $\vartriangleleft \hspace{3pt} \in\{<,\leq\}$, or $b>1$ and $\vartriangleleft \hspace{3pt} \in\{>,\geq\}$, let $\psi_i \dfn \bot$. If $b<0$ and $\vartriangleleft \hspace{3pt} \in \{>,\geq\}$, or $b>1$ and $\vartriangleleft \hspace{3pt} \in\{<,\leq\}$, let $\psi_i \dfn\top$. 
Finally, let $\varphi \dfn \bigsqcup_{j= 1..m} \bigwedge_{i\in I_j} \psi_{i}$.

Since the formulae $\Al_k$ 
describe disjoint events, the probability of $\bigvee_{s_k \in \B_\sigma \mid a_k^i = 1} \Al_k$
is the sum $a_1^i\epsilon_{s_1}^T + \dots + a_n^i\epsilon_{s_n}^T$ of the probabilities of the formulae $\Al_k$.
Thus, a causal multiteam $T$ of signature $\sigma$ satisfies $\psi_i$ iff $a_1^i\epsilon_1^T + \dots + a_n^i\epsilon_n^T \vartriangleleft b^i$; and thus, $T\models\varphi$ iff $\pvec_T\in\Pvec_\K$. Since $\K$ is closed under change of laws and rescaling, and it has the empty multiteam property, 
it is the class of all causal multiteams of signature $\sigma$ being empty or having probability vector in $\Pvec_\K$. Thus $\K=\K_\varphi$.
\qed
\end{proof}

\begin{restatable}{theorem}{Pchar}\label{thm: characterization of P}
A class $\K$ of multiteams of signature $\sigma$ is definable
in $\PP$ if and only if:
\begin{enumerate}
\item $\K$ is closed under change of laws
\item $\K$ is closed under rescaling
\item $\K$ is signed monic.
\end{enumerate} 
\end{restatable}

\begin{proof} 
   Left-to right. The fact that $\PP$ has the empty multiteam property and is closed under rescaling follows from Theorem  \ref{theorem: expressivity of PCOinf};
   closure under change of laws follows from Lemma \ref{lemma:trivial laws}.
The fact that a set of models $\K$ defined in $\PP$ is signed monic is given by Lemma \ref{lemma: P is signed monic}.

The right-to-left entailment is proved via a direct translation from finite unions of finite systems of signed monic inequalities into $\PP$ formulae. Suppose that $\K$ has all the four listed properties. Since it is signed monic, $\Pvec_\K$ is the union of the solution sets of a finite number $m$ of finite systems $A_1,\dots,A_m$ of 1) monic inequalities and 2) signed monic inequalities. Since we already know (by the proof of Theorem \ref{thm: characterization of P minus}) that monic inequalities are expressible in $\PP^- \subseteq \PP$, the proof can proceed as that of Theorem\ref{thm: characterization of P minus}, provided we show that signed monic inequalities with constant coefficient $0$ can be expressed in $\PP$. But since such inequalities are of the form $\sum_{i\in I}\epsilon_i - \sum_{j\in J}\epsilon_j \vartriangleleft 0$ with $I\cap J =\emptyset$, they are easily translated as $\Pr(\bigvee_{i\in I}\Al_i)\vartriangleleft \Pr(\bigvee_{j\in J}\Al_j)$. \qed
\end{proof}

It is not immediate to see whether $\PP^- \leq \PP$ is strict. However, by analyzing the geometry of $\Delta^{n-1}$ we are to show that there are signed monic classes of causal multiteams that are not monic.
The following lemma establishes that not all signed monic probability sets can be captured by monic inequalities; more specifically, that this happens for a set defined by a single signed monic inequality. \corrf{Luckily, we can build such a counterexample already in the $\Delta^2$ simplex in 3-dimensional space; this environment for probability vectors is realized when one uses a signature with exactly one variable, which can take three distinct values.} Together with the previous theorem, this implies that $\PP^- < \PP$.

\begin{lemma}\label{lemma: not all signed monic are monic}
Consider the probability set $\Pvec\subset \Delta^{2}$ which is defined by the homogeneous signed monic inequality $\epsilon_1 - \epsilon_2 \leq 0$ (together with the simplex constraint $\epsilon_1 +\epsilon_2+\epsilon_3 = 1$). Then, $\Pvec$ is not a monic probability set.\looseness=-1
\end{lemma}

\begin{proof}
    In tridimensional space, the equation $\epsilon_1 - \epsilon_2 = 0$ (the surface of the subspace defined by $\epsilon_1 - \epsilon_2 \leq 0$) defines the plane $D$ that is orthogonal to the $(\epsilon_1,\epsilon_2)$ plane and that bisects the first and third quadrant of the plane $(\epsilon_1,\epsilon_2)$. The intersection of $D$ with the simplex $\Delta^{2}$ is then the line that passes through the vertex $\epsilon_1 = 0, \epsilon_2 = 0,\epsilon_3 = 1$ of the simplex and is orthogonal to the opposite side of the simplex (the projection of this line on the $(\epsilon_1,\epsilon_2)$ plane is depicted, as a thick line, in Figure \ref{figure: P undefinable in P minus}).  

    We will show that this line, call it $L$, cannot be contained in the (1-dimensional) boundary\footnote{I.e., the perimeter.} of any monic probability set. Thus, in particular, $L$ is not a monic probability set.

    First of all, we show that the intersection of the simplex with the space defined by \emph{one} monic inequality cannot include $L$ as a subset. First of all, monic inequalities of the form $\epsilon_1 +\epsilon_2+\epsilon_3 = b$, once intersected with the simplex $\epsilon_1 +\epsilon_2+\epsilon_3 = 1$, will have either empty boundary (when $b\neq 1$) or the boundary will be the perimeter of the simplex (when $b=1$); in neither case the boundary includes $L$ as a subset. We then observe that we do not really need to consider inequalities with two variables, because, e.g., the system   
    \[
\left\{
\begin{array}{l}
  \epsilon_2 + \epsilon_{3}  \vartriangleleft b     \\
  \epsilon_1 + \epsilon_2 + \epsilon_3 = 1      
\end{array}
\right.
\]
is equivalent (by replacing $\epsilon_2 + \epsilon_{3}$ with $1-\epsilon_3$ in the first formula) to the system 
\[
\left\{
\begin{array}{l}
  \epsilon_{3}  \vartriangleright 1- b     \\
  \epsilon_1 + \epsilon_2 + \epsilon_3 = 1      
\end{array}
\right.
\]
 where the inequality only contains one variable. Let us then consider the boundaries of the spaces defined by monic inequalities with one variable. These are characterized by (the intersections with the simplex of) equalities of the forms $\epsilon_1 =b$, $\epsilon_2 =b$ or $\epsilon_3 =b$. The equation $\epsilon_1 =b$ describes a plane parallel to the $(\epsilon_2,\epsilon_3)$ plane; thus, its intersection with the simplex (if not empty) is a line that is parallel to the side of the simplex that is included in the  $(\epsilon_2,\epsilon_3)$ plane. Analogously, $\epsilon_2 =b$ or $\epsilon_3 =b$ describe lines parallel to the other two sides of the simplex. But then, since $L$ is not parallel to any side of the simplex, it is not included in any of the lines described by these equations. Thus, $L$ is not the boundary of a space described by a single monic inequality.

Now, if $F$ is the boundary of the space defined by a \emph{system} of monic inequalities, say $e_1,\dots,e_n$, then it consists of (parts of) the perimeter of a polygon (inside the simplex). Each side $S$ of this polygon is a subset of a line defined by an equality of the form $\epsilon_i=b$, for $i=1,2$ or $3$. But, we have already seen that the spaces defined by these equations are not parallel to $L$; thus, each side $S$ intersects $L$ in at most one point $p_S$. Since there is only a finite number of sides, we obtain a finite number of points $p_S$, and $L$ cannot be a subset of the union of a finite number of points.     

Finally, suppose $F$ is the boundary of a \emph{finite union} of polygons $P_1,\dots P_m$ defined by finite systems of monic inequalities. Each of these polygons $P_j$ has a perimeter $F_j$. But then, it is clear that $F\subseteq \bigcup_{j\in J} F_j$. Since, as we have seen above, each of the $F_j$ may intersect $L$ at most in a finite number of points, and there are finitely many polygons, we conclude that also $F$ intersects $L$ at most in a finite number of points. Thus, $L\not\subseteq F$. \qed
\end{proof}



\begin{figure}
    \centering
    \begin{tikzpicture}[scale=2.5]
        \draw[->] (-0.2, 0) -- (1.2, 0) node[right] {$\epsilon_1$};
        \draw[->] (0, -0.2) -- (0, 1.2) node[above] {$\epsilon_2$};

        \filldraw[fill=gray!20, draw=black] (0, 0) -- (0, 1) -- (1, 0) -- cycle;


        \draw[line width = 0.8mm] (0,0) -- (0.5,0.5);
        \draw[dashed] (0.3,0) -- (0.3,0.7);
        \draw[dashed] (0,0.1) -- (0.9,0.1);
        \draw[dashed] (0.8,0) -- (0,0.8);

        \node[below left] at (0, 0) {(0, 0)};
        \node[above] at (-0.22, 0.9) {(0, 1)};
        \node[below right] at (0.8, 0) {(1, 0)};
    \end{tikzpicture}
    \caption{Projection of the standard simplex in the $(\epsilon_1,\epsilon_2)$-plane. The thick line is the frontier of the object defined by the inequality in Lemma \ref{lemma: not all signed monic are monic}. The dashed lines exemplify frontiers of $\PP^-$-definable sets, always parallel to the sides of the triangle.}
    \label{figure: P undefinable in P minus}
\end{figure}

\begin{remark}
    \corrf{Lemma \ref{lemma: not all signed monic are monic} must not be misinterpreted as saying that \emph{any} probability set that is defined merely in terms of homogeneous signed monic inequalities is not a monic set. } There are systems of non-monic inequalities that are equivalent to systems of monic inequalities. For example, the system
    \[
    \left\{\begin{array}{l}
    \epsilon_1 - \epsilon_2 >0  \\
    \epsilon_2 \geq 0  \\
    \epsilon_2 \leq 0 
\end{array}
\right.
    \]
    is equivalent to  the monic system
    \[
    \left\{\begin{array}{l}
    \epsilon_1 >0  \\
    \epsilon_2 \geq 0  \\
    \epsilon_2 \leq 0 
\end{array}
\right.
    \]
    \corrf{Moreover, in some cases even a space defined by \emph{one} signed monic inequality (featuring both $1$ and $-1$ coefficients) can, in fact, be monic. This is \emph{always} the case for signed monic inequalities that have no null variable coefficients.\footnote{We thank one of the anonymous reviewers for this insight.} Indeed, such an inequality is of the form $\sum_{i\in I}\epsilon_i - \sum_{j\in J}\epsilon_j \vartriangleleft b$, where $I\cap J= \emptyset$ and $n = |I\cup J|$ is the dimension of the space in which the appropriate simplex, $\Delta^{n-1}$, is embedded. Now, in   $\Delta^{n-1}$ the further constraint $\sum_{k= 1..n}\epsilon_k =1$ holds. If either $I$ or $J$ is empty, then the inequality is monic. Assuming they are both nonempty, we may pick an $m\in J$. Since we have $\epsilon_m = 1-\sum_{(I\cup J)\setminus \{m\}}$, we may eliminate $\epsilon_m$ from the initial inequality, obtaining the inequality $2\sum_{i\in I}\epsilon_j \vartriangleleft b$, or $\sum_{i\in I}\epsilon_j \vartriangleleft \frac{b}{2}$, which is monic.}
\end{remark}

Next we turn to characterizing the expressivity of $\PC$. First note that while $\PC$ is in general more expressive than $\PP$ (Corollary \ref{cor: easy comparisons}), if we restrict attention to causal multiteams with a fixed function component, all occurrences of $\cf$ can be eliminated from $\PC$ formulae (or even $\PCO$ formulae). 

\begin{restatable}{proposition}{ignoringcf}\label{proposition: ignoring cf}
Let $\varphi\in \PC_{\sigma}$ (resp. $\PCO_{\sigma}$),
 and $\F$ a function component of signature $\sigma$. Then there is a formula $\varphi^\F\in \PP_{\sigma}$ (resp. $\PO_{\sigma}$)
  such that, for every causal multiteam $T$ of signature $\sigma$ and function component $\F$,
\(
T\models \varphi \Leftrightarrow T \models \varphi^\F.
\)
\end{restatable}

\begin{proof}
Write $\alpha_s$ for the formula $\SET W = s(\SET W)$. 
First, for every subformulae of $\varphi$ of the form $\beta\supset\psi$, replace $\beta$ with $\bigvee_{(\{s\},\F) \models \beta}\alpha_s$ (this removes occurrences of $\cf$ from antecedents of $\supset$).  Next, we use the fact that $\cf$ distributes over $\land,\sqcup,\supset$, and the fact that nested counterfactuals are equivalent to unnested ones\footnote{It can be shown that $\SET X = \SET x \cf (\SET Y = \SET y \cf \psi)$ is equivalent to $(\SET X' = \SET x' \land \SET Y = \SET y) \cf \psi)$, where $\SET X' := \SET X \setminus \SET Y$ and $\SET x' := \SET x \setminus \SET y$.}, to guarantee that the consequents of $\cf$ are atoms. The atoms can be assumed to be probabilistic (since $X=x \equiv \Pr(X=x)\geq 1$, and similarly for $X\neq x$). Then, we use the equivalences  
\[
\SET X = \SET x \cf \Pr(\alpha)\vartriangleleft \epsilon   \equiv    \Pr(\SET X = \SET x \cf \alpha) \vartriangleleft \epsilon 
\]
\[
\SET X = \SET x \cf   \Pr(\alpha)\vartriangleleft \Pr(\beta)   \equiv    \Pr(\SET X = \SET x \cf \alpha) \vartriangleleft  \Pr(\SET X = \SET x \cf \beta)
\]
to ensure that all the occurences of $\cf$ are inside arguments of $\Pr$.\footnote{The equivalences can be easily seen to hold using the fact that interventions preserve the multiplicity of assignments in a causal multiteam.} Finally, we replace each subformula of the form $\Pr(\alpha) \vartriangleleft \epsilon$ with $\Pr(\bigvee_{(\{s\},\F)\models \alpha} \alpha_s ) \vartriangleleft \epsilon$, and similarly for comparison atoms.  We call $\varphi^\F$ the resulting formula. The equivalence $T\models \varphi \Leftrightarrow T \models \varphi^\F$ can then be proved by an inductive argument.\qed 
\end{proof}

\noindent Notice that, for any fixed finite signature $\sigma$, there is only a finite number of distinct function components. We denote the set they form as $\mathbb F_\sigma$.

\begin{restatable}{theorem}{charPC}\label{thm: characterization of PC}
Let $\K$ be a class of causal multiteams of signature $\sigma$. $\K$ is definable by a $\PC_{\sigma}$ formula if and only if: 
\begin{enumerate}
\item $\K$ has the empty multiteam property
\item $\K$ is closed under rescaling 
\item $\K = \bigcup_{\F\in\mathbb{F}_\sigma}\K^\F$, where  each $\K^\F$ is a signed monic set of causal multiteams of function component $\F$.
\end{enumerate}
\end{restatable}

\begin{proof}
We have already mentioned that there is a $\PCO$ formula $\Phi^\F$ characterizing the property of having function component $\F$ (Theorem \ref{thm: phif} from the Appendix). We can obtain an equivalent formula (call it $\Psi^\F$) in $\PC$ 
 by replacing each subformula of $\Phi^\F$ of the form $\alpha \supset \beta$ with $\Pr(\alpha^d\lor \beta) =1$ (the trick works because, first, there is no occurrence of $\supset$ in $\alpha^d$, and, secondly, no consequent of $\supset$ in $\Phi^\F$ contains probabilistic atoms).

$\Rightarrow$) Suppose $\K = \K_\varphi$, where $\varphi\in\PC_{\sigma}$. Now 
define, for each $\F\in \mathbb{F}_\sigma$, $\K^\F := \K_{\varphi \land \Psi^\F}$, where $\Psi^\F$ is as described above. Clearly $\varphi \equiv \bigsqcup_{\F\in \mathbb{F}_\sigma} (\varphi \land \Psi^\F)$, so $\K_\varphi = \bigcup_{\F\in\mathbb{F}_\sigma}\K^\F$.

Now, by Theorem \ref{theorem: expressivity of PCOinf}, $\K_\varphi$ is closed under rescaling  and  has the empty multiteam property. 
Next, observe that, by Proposition \ref{proposition: ignoring cf}, 
for every $\F\in\mathbb{F}_\sigma$ there is a formula of $\PP_{\sigma}$, call it $\varphi^\F$, which is satisfied by the same causal multiteams of function component $\F$ as $\varphi \land \Psi^\F$ is. In other words, $\K^\F$ is the restriction of $\K_{\varphi^\F}$ to causal multiteams of function component $\F$. Thus, since $\K_{\varphi^\F}$ is closed under change of laws (Lemma \ref{lemma:trivial laws}), we have $\Pvec_{\K^\F} = \Pvec_{\K_{\varphi^\F}}$. Now $\K_{\varphi^\F}$ is signed monic (Theorem  \ref{thm: characterization of P}), and thus by $\Pvec_{\K^\F} = \Pvec_{\K_{\varphi^\F}}$ we conclude that also $\K^\F$ is signed monic.

$\Leftarrow$) Suppose $\K$ is closed under rescaling, has the empty multiteam property and $\K = \bigcup_{\F\in\mathbb{F}_\sigma}\K^\F$ for some sets $\K^\F$ as in the statement. Write $\hat \K^\F$ for the set of all causal multiteams of signature $\sigma$ whose team component appears in $\K^\F$. It is straightforward then that also $\hat \K^\F$ is closed under rescaling, has the empty multiteam property and is signed monic; however, $\hat \K^\F$ is also, by definition, closed under change of laws. Thus, by Theorem \ref{thm: characterization of P}, there is a $\PP$ formula $\varphi^\F$ such that $\hat \K^\F = \K_{\varphi^\F}$. Note that, $\K^\F$ is the set of all causal multiteams of $\K_{\varphi^\F}$ that have function component $\F$. Thus $\K^\F = \K_{\varphi^\F \land \Psi^\F}$. Thus $\K$ is defined by the $\PC_{\sigma}$ formula $\bigsqcup_{\F\in\mathbb{F}_\sigma}(\varphi^\F \land \Psi^\F)$.\qed
\end{proof}

\noindent Note that the sets $\K^\F$ in the statement of the theorem are themselves closed under rescaling if $\K$ is. This immediately follows from the fact that any two causal multiteams $(T,\F),(S,\G)$ with $\F\neq\G$ are \emph{not} rescalings of each other.

\subsection{Signed binary probability sets: $\PO$ and $\PCO$}\label{subs: PO and PC}

A peculiarity of languages such as $\PO$ and $\PCO$, which feature the operator $\supset$, is that they allow to discuss conditional probabilities. Writing $P_T(\beta \mid \alpha)$ for the conditional probability, in $T$, of $\beta$ given $\alpha$ (that is, $P_T(\beta \mid \alpha):= \frac{P_T(\alpha\land\beta)}{P_T(\alpha)}$, provided $P_T(\alpha)>0$), we have the following.

\begin{proposition}[\cite{BarSan2024}, Theorem 6.1]\label{prop: conditional probabilities}
Let $T$ be a nonempty causal multiteam and $\vartriangleright \in \{<, \leq, >, \geq, =\}$. Then:
    \begin{enumerate}
        \item $T\models \alpha \supset \Pr(\beta) \vartriangleright \epsilon$ iff $P_T(\alpha)\leq 0$ or $P_T(\beta) \vartriangleright \epsilon$.
        \item $T\models \alpha \supset \Pr(\beta) \vartriangleright \Pr(\beta)$ iff $P_T(\alpha)\leq 0$ or $P_T(\beta) \vartriangleright P_T(\beta)$.
    \end{enumerate}
\end{proposition}

  \begin{definition}\label{def:signed_binary}
    A probability set is \textbf{signed binary} if it is defined, over $\Delta^{n-1}$, as a finite union of probability sets defined by finite systems of inequalities of the following two types:
       \begin{enumerate}
       \item monic inequalities 
       \item homogeneous signed binary inequalities. 
       \end{enumerate}
    \end{definition}


\begin{lemma}\label{lemma: PO is linear}
Every formula $\varphi\in\PO$ is signed binary.
\end{lemma}

\begin{proof}
The proof proceeds by induction on $\varphi$. We only discuss the most difficult case, when $\varphi$ is of the form $\alpha\supset \psi$. Write $\vartriangleleft$ for any symbol in $\{ \leq,\geq,<,>\}$. 
Using the distributivity of $\supset$ over $\land$ and $\lor$, and the equivalences $X=x \equiv \Pr(X=x) = 1, X \neq x \equiv \Pr(X \neq x) = 1, \SET X = \SET x \cf \Pr(\alpha)\vartriangleleft \epsilon \equiv \Pr(\SET X = \SET x \cf \alpha)\vartriangleleft\epsilon$ and $\SET X = \SET x \cf \Pr(\alpha)\vartriangleleft \Pr(\beta) \equiv \Pr(\SET X = \SET x \cf \alpha)\vartriangleleft\Pr(\SET X = \SET x \cf \beta)$, we can assume $\psi$ to be a probabilistic atom. Hence we have two cases.

 1) Assume $\psi$ is $\Pr(\beta)\vartriangleleft b$. 
 Now $T=(T^-,\F)\in\K_\varphi$ iff (by Proposition \ref{prop: conditional probabilities}) either $P_T(\alpha) \leq 0$ or $P_T(\beta\mid\alpha)\vartriangleleft b$. The latter is equivalent to $P_T(\beta\land \alpha)\vartriangleleft b \cdot P_T(\alpha)$, which can be rewritten as
 \[
 \sum_{\substack{s\in\B_\sigma \\ \{s\}\models\beta\land\alpha}}\epsilon^T_s \vartriangleleft b \cdot \sum_{\substack{s\in\B_\sigma \\ \{s\}\models\alpha}}\epsilon^T_s
 \]
where we write e.g. $\{s\}\models\alpha$  as a shorthand for $(\{s\},\F)\models\alpha$.
 
 The above can be rewritten as
  \[
 \sum_{\substack{s\in\B_\sigma \\ \{s\}\models\beta\land\alpha}}\epsilon^T_s \vartriangleleft b \cdot \big( \sum_{\substack{s\in\B_\sigma \\ \{s\}\models\beta \land \alpha}} \epsilon^T_s + \sum_{\substack{s\in\B_\sigma \\ \{s\}\models\neg \beta \land \alpha}} \epsilon^T_s \big)
 \]
which again is equivalent to
\begin{equation}\label{eq:1}
 (1-b)\cdot\sum_{\substack{s\in\B_\sigma \\ \{s\}\models\beta\land\alpha}}\epsilon^T_s + (- b) \cdot \sum_{\substack{s\in\B_\sigma \\ \{s\}\models \neg \beta \land\alpha}}\epsilon^T_s \vartriangleleft  0.
 \end{equation}
 Now, since $b\in [0,1]$, we have $1-b \geq 0$ and $-b\leq 0$. Then, by multiplying both sides of \eqref{eq:1} by a common denominator of $1-b$ and $-b$, we obtain a \corrf{homogeneous} signed binary inequality.

 On the other hand, the inequality $P_T(\alpha)\leq 0$ can be rewitten as $ \sum_{\{s\}\models\alpha}\epsilon_s\leq 0$. 
Thus $\Pvec_\varphi$ is the union of two sets defined by 
\corrf{homogeneous} signed binary inequalities.

2) Assume $\psi$ is $\Pr(\beta)\vartriangleleft \Pr(\gamma)$. Now $T\in\K_\varphi$ iff either $P_T(\alpha)\leq 0$ or $P_T(\beta\mid\alpha)\vartriangleleft P_T(\gamma\mid\alpha)$. The proof then proceeds as in the previous case. 
\qed
\end{proof}


\begin{restatable}{theorem}{POchar}\label{thm: characterization of PO}
A class $\K$ of multiteams of signature $\sigma$ is definable
in $\PO$ if and only if:
\begin{enumerate}
\item $\K$ is closed under change of laws
\item $\K$ is closed under rescaling, and
\item $\K$ is signed binary.
\end{enumerate} 
%
\end{restatable}

\begin{proof}
$\Rightarrow$) By Theorem \ref{theorem: expressivity of PCOinf}, $\K$ is closed under rescaling. Closure under change of laws follows from Lemma \ref{lemma:trivial laws}. 
 Lemma \ref{lemma: PO is linear} shows that $\Pvec_\K$ is signed binary. 
 The empty multiteam property is given by Theorem \ref{thm: CO flatness}.

$\Leftarrow$) 
The proof strategy is analogous to that used for the characterization of $\PP$ (in Theorem \ref{thm: characterization of P}), although it involves more difficult calculations. 
%
%
%
%
%
%
 We need to show that every constraint of the form 
\[
c^- \sum_{i\in I} \epsilon_i + c^+ \sum_{j\in J} \epsilon_j \vartriangleleft 0
\]
where $I\cap J = \emptyset$, $c^-,c^+ \in \mathbb Z$, $c^- \leq 0$, $c^+ \geq 0$, 
can be expressed in $\PO$.

Write $d$ for $c^+ -c^-$.  Notice that $-d \leq c^- \leq 0 \leq c^+ \leq d$. We can also assume that $d>0$ (the case when $d=0$ is covered by Theorem \ref{thm: characterization of P minus}). Then $-\frac{c^-}{d}$ is a rational number in $[0,1]$, and thus the following is a $\PO$ formula (where, as before, $\Al_j$ stands for $\SET W = s_j(\SET W)$):
\[
\Big(\bigvee_{k\in I\cup J} \Al_k\Big) \supset  \Pr(\bigvee_{j\in J}\Al_j)\vartriangleleft -\frac{c^-}{d}.
\]
Now we have
\begin{align*}
 & T\models \Big(\bigvee_{k\in I\cup J} \Al_k\Big) \supset  \Pr(\bigvee_{j\in J}\Al_j) \vartriangleleft -\frac{c^-}{d} \\
 & \iff P_T(\bigvee_{j\in J}\Al_j  \mid  \bigvee_{k\in I\cup J} \Al_k )\vartriangleleft -\frac{c^-}{d} \\
 & \iff d\cdot P_T(\bigvee_{j\in J}\Al_j \land  \bigvee_{k\in I\cup J} \Al_k) \vartriangleleft - c^- \cdot P_T(\bigvee_{k\in I\cup J}\Al_k) \\
  & \iff d\cdot P_T(\bigvee_{j\in J}\Al_j) \vartriangleleft - c^- \cdot P_T(\bigvee_{k\in I\cup J}\Al_k) \\
  & \iff d\sum_{j\in J}\epsilon_{s_j}^T \vartriangleleft -c^- \sum_{k\in I\cup J}\epsilon_{s_k}^T \\
  & \iff c^- \sum_{i\in I}\epsilon_{s_i}^T + (d+c^-)\sum_{j\in J}\epsilon_{s_j}^T \vartriangleleft 0 \\
  & \iff c^- \sum_{i\in I}\epsilon_{s_i}^T + c^+\sum_{j\in J}\epsilon_{s_j}^T \vartriangleleft 0, 
\end{align*}
as required.
\qed
\end{proof}

In order to prove that $\PO$ is strictly more expressive than $\PP$ , we can follow a similar strategy as for separating $\PP$ and $\PP^-$. In other words, we use Theorem \ref{thm: characterization of PO} together with the fact that there are signed binary probability sets that are not signed monic, as established by the following lemma.



\begin{lemma}\label{lemma: not all signed binary are signed monic}
Consider the probability set $\Pvec\subset \Delta^{2}$ which is defined by the homogeneous signed binary inequality $\epsilon_1 - 3 \epsilon_2 \leq 0$ (together with the simplex constraint $\epsilon_1 +\epsilon_2+\epsilon_3 = 1$). Then, $\Pvec$ is not a signed monic probability set.\looseness=-1
\end{lemma}

    \begin{proof}
        In tridimensional space, the equation $\epsilon_1 - 3\epsilon_2 = 0$ (the equation of the surface of the subspace defined by $\epsilon_1 - 3\epsilon_2 \leq 0$) defines a plane $D$ that passes through the $\epsilon_3$ axis and crosses the first and third quadrants of the plane $(\epsilon_1,\epsilon_2)$, without bisecting them (the projection of this line on the $(\epsilon_1,\epsilon_2)$ plane is depicted, as a thick line, in Figure \ref{figure: PO undefinable in P}).
        %
        The intersection of $D$ with the simplex $\Delta^{2}$ is then the line that passes through the vertex $\epsilon_1 = 0, \epsilon_2 = 0,\epsilon_3 = 1$ of the simplex and is orthogonal to the opposite side of the simplex. 

    We will show that this line, call it $M$, cannot be contained in the (1-dimensional) boundary of any signed monic probability set. Thus, in particular, $M$ is not a signed monic probability set.

    First of all, we show that the intersection of the simplex with a signed monic probability set defined by \emph{one} inequality cannot include $M$ as a subset (remember that, by the definition of signed monic probability set, such an inequality is either monic or signed monic with constant coefficient $0$). In case the inequality is monic, this is proved by considering the same cases raised in the proof of Lemma 16. If the inequality is signed monic with constant coefficient $0$, then its boundary is the line $L: \epsilon_1 -\epsilon_2 = 0$ (note that the only other signed monic equality with constant coefficient $0$, namely $-\epsilon_1 +\epsilon_2 = 0$, defines the same line). Since the $(\epsilon_1,\epsilon_2)$-projection of $L$ bisects the first and third quadrant, while the projection of $M$ does not, $L$ and $M$ only meet in the vertex $\epsilon_1 = 0, \epsilon_2 = 0,\epsilon_3 = 1$. Thus, $M\not\subseteq L$. 
    
    The proof then proceeds exactly as in Lemma \ref{lemma: not all signed monic are monic}.\qed
    \end{proof}



\begin{figure}
    \centering
    \begin{tikzpicture}[scale=2.5]
        \draw[->] (-0.2, 0) -- (1.2, 0) node[right] {$\epsilon_1$};
        \draw[->] (0, -0.2) -- (0, 1.2) node[above] {$\epsilon_2$};

        \filldraw[fill=gray!20, draw=black] (0, 0) -- (0, 1) -- (1, 0) -- cycle;


        \draw[line width = 0.7mm] (0,0) -- (0.75,0.25);
        \draw[dashed] (0,0) -- (0.5,0.5);
        \draw[dashed] (0.3,0) -- (0.3,0.7);
        \draw[dashed] (0,0.1) -- (0.9,0.1);
        \draw[dashed] (0.8,0) -- (0,0.8);

        \node[below left] at (0, 0) {(0, 0)};
        \node[above] at (-0.22, 0.9) {(0, 1)};
        \node[below right] at (0.8, 0) {(1, 0)};
    \end{tikzpicture}
    \caption{Projection of the standard simplex in the $(\epsilon_1,\epsilon_2)$-plane. The thick line is the frontier of the object defined by the inequality in Lemma \ref{lemma: not all signed binary are signed monic}. The dashed lines exemplify frontiers of $\PP$-definable sets, i.e. lines parallel to the sides of the triangle or the line orthogonal to the $(0,1),(1,0)$ side and passing through the origin.}
    \label{figure: PO undefinable in P}
\end{figure}

\noindent Actually, the lemma immediately yields multiple separation results.

\begin{restatable}{proposition}{POgreaterthanP}\label{prop: PO greater than P}
1) $\PP < \PO$, \quad 2)  $\PO\not\leq\PC$, \quad 3) $\PC <\PCO$.  
\end{restatable}

\begin{remark}
    \corrf{Again, we can see that there are nontrivial systems of signed binary inequalities that define signed monic probability sets. 
    For example, the following system}
    \[
\left\{\begin{array}{lcr}
\epsilon_1+2\epsilon_2+2\epsilon_3 & \geq & 0 \\
\epsilon_2+\epsilon_3 -\epsilon_4 & \leq & 0 \\
\epsilon_2+\epsilon_3 -\epsilon_4 & \geq & 0
\end{array}
\right.
\]
is equivalent to 
\[
\left\{\begin{array}{lcr}
\epsilon_1+\epsilon_2+\epsilon_3+\epsilon_4 & \geq & 0 \\
\epsilon_2+\epsilon_3 -\epsilon_4 & \leq & 0 \\
\epsilon_2+\epsilon_3 -\epsilon_4 & \geq & 0
\end{array}
\right.
\]
which \corrf{defines a signed monic set}.
\end{remark}

We are finally ready to characterize the expressive power of $\PCO$.
%
\begin{restatable}{theorem}{PCOchar}\label{thm: characterization of PCO}
Let $\K$ be a class of causal multiteams of signature $\sigma$. $\K$ is definable by a $\PCO_{\sigma}$ formula if and only if:
\begin{enumerate}
\item it has the empty multiteam property
\item it is closed under rescaling
\item $\K = \bigcup_{\F\in\mathbb{F}_\sigma}\K^\F$, where  each $\K^\F$ is a signed binary set of causal multiteams of function component $\F$.
\end{enumerate}
\end{restatable}

\begin{proof}
The argument is very similar as in the proof of Theorem \ref{thm: characterization of PC}, using Lemma \ref{lemma: PO is linear} instead of Theorem \ref{thm: characterization of P}.
\qed
\end{proof}

\vspace{5pt}



By Theorem \ref{theorem: expressivity of PCOinf}, $\PCO^\omega$ formulae may characterize arbitrary probability sets. By Theorem \ref{thm: characterization of PCO}, instead, we know that the probability sets of $\PCO$
 formulae are all definable in terms of linear inequalities. A strict inclusion of languages immediately follows. An alternative proof for this using a counting argument was given in \cite{BarSan2024}.
 \begin{corollary}\label{cor: PCO smaller than PCOinf}
 $\PCO < \PCO^\omega$.    
 \end{corollary}

\section{Expressive power of languages with the strict tensor} \label{sec: tensor}

In this section we analyze the expressive power of languages that allow free use of the \emph{strict tensor} operator $\lor$. We will focus on the language $\PCOs$ and its fragments $\Psm$ and $\Ps$, as described in the preliminary section. 
To see that these are the only languages with strict tensor that are worth considering for semantic classification purposes, it suffices to observe that the selective implication $\supset$ is definable in terms of $\lor$:
\begin{itemize}
    \item $\alpha \supset \psi$ is equivalent to $\alpha^d \lor \psi$
\end{itemize}
where we notice, importantly, that $\alpha^d$ is a formula without occurrences of $\supset$. 
Thus, $\PO \leq \Ps$. It then immediately follows that the extension of $\PO$ with $\lor$ is equiexpressive to $\Ps$, and that the extension of $\PC$ with $\lor$ is equiexpressive with $\PCOs$.  




It is easy to prove that $\Psm, \Ps$ and $\PCOs$ have the empty multiteam property. It is more subtle to prove that they are closed under rescaling. Let us say something more about rescalings. It is easy to see that, for nonempty causal multiteams $S,T$, $S$ is a rescaling of $T$ ($S\sim T$) if and only if there is a positive, nonzero rational number $q$ such that $\#(s,T) = q \cdot \#(s,S)$ for each $s\in \B_\sigma$ (this, in particular, tells us that $\frac{|T^-|}{|S^-|} = q$).  Since clearly such number is unique, we can also write this relation as $T= q S$. In case $q\in \mathbb N$, we say that $T$ is a \textbf{multiple} of $S$. If $T = mR$ and $T=nS$, where $m,n\in \mathbb{N}^+$, we say $T$ is a \textbf{common multiple} of $R$ and $S$. It is not difficult to show (and the details can be found in \cite{BarSan2024}, Lemma A.4) that if $R\sim S$, then they have a common multiple. The same terminology and considerations may be applied to multiteams $R^-,S^-,T^-$.


\begin{lemma}[Closure under rescaling]\label{lemma: rescaling for PCOs}
Let $S=(S^-,\F), T=(T^-,\F)$ be causal multiteams of signature $\sigma$, $S \sim T$, $\varphi\in\PCOs_{\sigma}$, and $S\models \varphi$. Then $T\models \varphi$.
\end{lemma}

\begin{proof}

By induction on $\varphi$. All cases are taken care of by the proof of Lemma A.7 
from \cite{BarSan2024}, with the exception of the case for $\lor$, which we provide here.

Let $\varphi$ be of the form $\psi \lor \chi$. Let $S\models \psi \lor \chi$. Then there are $S_1 = (S_1^-,\F),S_2 = (S_2^-,\F)$ causal submultiteams of $S$ such that $S_1^- \cup S_2^- = T^-$, $S_1^- \cap S_2^- =\emptyset$, $S_1\models \psi$ and $S_2\models \chi$. Now, since $S \sim T$, by the observations above 
$S$ and $T$ have a common multiple, say $R = (R^-,\F)$. In particular,  $R = mS$ for some $m\in \mathbb{N}^+$. 

Now define a subset $R_1^-$ of $R^-$ by picking $m\cdot \#(s,S_1^-)$ copies of each $s\in Team((S_1)^-)$ and another $R_2^-\subseteq R^-$ by picking, from $R^-\setminus R_1^-$, $m\cdot \#(s,S_2^-)$  copies of each $s\in Team((S_2)^-)$. 
By 
the considerations above, $(S_1^-,\F) \sim (R_1^-,\F)$ and $(S_2^-,\F) \sim (R_2^-,\F)$. 


By definition, $R_1^- \cap R_2^- = \emptyset$; we can also show that $R_1^- \cup R_2^- = R^-$. In order to do so, we show that these two sets contain the same number of copies for each assignment allowed by the signature. 
Indeed, if $s\in \B_\sigma$, 
$\#(s,R^-) = m \cdot \#(s, S^-) = m \cdot (\#(s,S_1^-) +  \#(s,S_2^-) ) = m \cdot \#(s,S_1^-) + m \cdot \#(s,S_2^-)$  
(where in the second equality we used $S^- = S_1^- \cup S_2^-$ and $S_1^- \cap S_2^- =\emptyset$), and by the definition of $R_1^-,R_2^-$ the latter is equal to $\#(s,R_1^-) +  \#(s,R_2^-) = \#(s,R_1^-\cup R_2^-)$. 

Now we can use $R_1^-$ and $R_2^-$ to define, in an analogous way, two disjoint subsets $T_1^- \sim R_1^-\sim S_1^-$ and $T_2^- \sim R_2^- \sim S_2^-$ of $T^-$ such that $T_1^- \cup T_2^- = T^-$. By the inductive hypothesis, $(T_1^-,\F)\models \psi$ and $(T_2^-,\F)\models \chi$. Thus $T\models \psi\lor\chi$. \qed
\end{proof}


\noindent The following corollary then follows immediately from Theorem \ref{theorem: expressivity of PCOinf}.

\begin{corollary}\label{cor: PCSOs in PCOomega}
    $\Ps \leq \PCOs < \PCO^\omega$.
\end{corollary}




We want to point out that the strict tensor operator has an important property related to inequalities: under a minimum of assumptions, it preserves linearity -- by producing, specifically, the convex hull of the probability sets of the formulas it is applied to. 
The \textbf{convex hull} of a set $\Pvec\subseteq \Delta^{n-1}$ is the set of all linear combinations $a\pvec_1+b\pvec_2$ where $\pvec_1,\pvec_2\in\Pvec$ and $a,b\in \mathbb Q\cap[0,1]$, $a+b=1$. It can be seen to be itself a subset of $\Delta^{n-1}$.

\begin{lemma}\label{lemma: strict tensor as convex hull}
Let $\psi,\chi$ be formulas such that are closed under rescaling, and suppose furthermore that either  $\psi$ or $\chi$ is closed under change of laws. Then $\Pvec_{\psi\lor\chi}$ is the convex hull of $\Pvec_\psi \cup \Pvec_\chi$.    
\end{lemma}

\begin{proof}
We prove the statement in case $\chi$ is closed under change of laws; the case for $\psi$ is analogous.

In one direction, suppose $\pvec\in\Pvec_{\psi\lor\chi}$, i.e. $\pvec= \pvec_T$ for some $T\in \K_{\psi\lor\chi}$. Then there are $T_1, T_2\leq T$ with $T_1^-\cap T_2^-= \emptyset$, $T_1^- \cup T_2^- = T^-$, $T_1\models\psi$ and $T_2\models \chi$. For any $s\in\B_\sigma$, $\epsilon_s^T = \frac{\#(s,T_1) + \#(s,T_2)}{|T^-|} = \frac{|T_1^-|}{|T^-|} \epsilon_{s}^{T_1} + \frac{|T_2^-|}{|T^-|} \epsilon_{s}^{T_2}$. Since the terms $\frac{|T_1^-|}{|T^-|}$ and $\frac{|T_2^-|}{|T^-|}$ do not depend on $s$, we have $\pvec_T =  \frac{|T_1^-|}{|T^-|}\pvec_{T_1} + \frac{|T_2^-|}{|T^-|}\pvec_{T_2}$ (as vectors). Since $\frac{|T_1^-|}{|T^-|}+\frac{|T_2^-|}{|T^-|}=1$, this means that $\pvec_T$ is in the convex hull of $\{\pvec_{T_1}\}\cup \{\pvec_{T_2}\}$ -- and thus of $\Pvec_{\psi} \cup \Pvec_{\chi}$.

Conversely, suppose $\pvec$ is in the convex hull of $\Pvec_{\psi} \cup \Pvec_{\chi}$.
Then there are $T_1=(T_1^-,\F)\in\K_\psi, T_2=(T_2^-,\G)\in \K_\chi$ and $a,b\in[0,1]\cap\mathbb Q$ such that $a+b=1$ and $\pvec= a\pvec_{T_1}+b\pvec_{T_2}$. We want then to build a causal multiteam $T=(T^-,\F)$ that satisfies $\psi\lor\chi$ and such that $\pvec_T = \pvec$. 

Let $c= a \cdot |T_2^-|\cdot k$ and $d= b \cdot |T_1^-|\cdot k$, where $k$ is a common denominator of $a \cdot |T_2^-|$, $b \cdot |T_1^-|$. Thus defined, $c,d$ are natural numbers, so we can define $T^-$ as the multiteam that has $c\cdot \#(s,T_1)+ d\cdot \#(s,T_2)$ copies of assignment $s$ (for each $s\in \B_\sigma$). Now observe that the size of $T^-$ is then $c\cdot|T_1^-| + d\cdot|T_2^-| = a\cdot|T_2^-|\cdot|T_1^-|\cdot k + b\cdot|T_1^-|\cdot|T_2^-|\cdot k = k\cdot|T_1^-|\cdot|T_2^-| \cdot (a+b) = k\cdot|T_1^-|\cdot|T_2^-|$. Thus 
\begin{align*}
\epsilon_s^T &= \frac{c\cdot \#(s,T_1)+ d\cdot \#(s,T_2)}{k\cdot|T_1^-|\cdot|T_2^-|} = \frac{a \cdot |T_2^-|\cdot k\cdot \#(s,T_1)+ b \cdot |T_1^-|\cdot k\cdot \#(s,T_2)}{k\cdot|T_1^-|\cdot|T_2^-|} \\
&= a\cdot \frac{\#(s,T_1)}{|T_1^-|} + b\cdot \frac{\#(s,T_2)}{|T_2^-|} = a\epsilon_s^{T_1} + b \epsilon_s^{T_2}.
\end{align*}
Thus $\pvec_T = a \pvec_{T_1} + b \pvec_{T_2} = \pvec$, as needed. 

Now observe that $T^-$ is the disjoint union of two rescalings of $T_1^-$, resp. $T_2^-$ (call them $S_1^-,S_2^-$). 
Since $\K_\chi$ is closed under change of laws, we have that $(T_2^-,\F)\models \chi$; and then, since $\K_\psi, \K_\chi$ are closed under rescaling, $(S_1^-,\F)\models \psi$ and $(S_2^-,\F)\models\chi$. Thus $T\in \K_{\psi\lor\chi}$. 
We can then conclude that $\pvec_T = \pvec \in \Pvec_{\psi\lor\chi}$.    \qed
\end{proof}

We say that a formula $\varphi$ is \textbf{linear} if $\Pvec_{\K_\varphi}$ is a set of solutions to a finite union of finite systems of linear inequalities.

\begin{lemma}\label{lemma: strict tensor preserves linearity}
Suppose $\psi,\chi$ 
are closed under rescaling and linear, and moreover 
either $\psi$ or $\chi$ is closed under change of laws. 
Then, $\psi\lor\chi$ 
is linear.
\end{lemma} 

\begin{proof}
By the assumptions, we know from Lemma \ref{lemma: strict tensor as convex hull} that $\Pvec_{\psi\lor\chi}$ is the convex hull of $\Pvec_\psi \cup \Pvec_\chi$ (which, being a finite union of linear sets, is itself a linear set). We can then use the general fact that the convex hull of any linear set 
 is still a linear set. \qed 
\end{proof} 




\begin{corollary} \label{cor: PS is linear}
   If $\varphi\in\Ps$, then $\Pvec_\varphi$ is linear.
\end{corollary}

\begin{proof}
We already know (Theorem \ref{thm: characterization of P}) that the atomic formulas have linear probability sets. We must prove that the connectives preserve linearity. We already know that this holds for $\lor$ by Lemma \ref{lemma: strict tensor preserves linearity} (which can be applied since the $\Ps$ formulas are closed under rescaling, by Lemma \ref{lemma: rescaling for PCOs}; and under change of laws, by Lemma \ref{lemma:trivial laws}). The proof that $\land,\sqcup$ preserve linearity is analogous to the proof of Lemma \ref{lemma: preserve polytopes}. \qed  
\end{proof}

The ideas from the proof of lemma \ref{lemma: strict tensor preserves linearity} can be extended to show that the language $\Ps$ can ``capture'' all polytopes.


\begin{lemma}\label{lemma: polytopes as convex hulls}
Every polytope\footnote{Note that the definition of polytope we have given ensures that a polytope is a convex set.} of $\Delta^{n-1}$ is the convex hull of a finite set of points.
\end{lemma}

\begin{proof}
Let $C$ be a polytope. It is defined by a finite number of inequalities $i_1,\dots,i_l$. Consider the corresponding equalities $e_1,\dots,e_l$, together with the equations $e_{l+1},\dots e_{l+n}$ that define the facets of $\Delta^{n-1}$ of dimension $n-2$. Now let $p_1,\dots,p_k$ be all the points of $\Delta^{n-1}$ that are intersections of spaces defined (in $\Delta^{n-1}$) by $n-1$ equations from $e_1,\dots, e_{l+n}$.\footnote{There may be $n-1$ tuples of such equations that have as intersection not just a point, but a space of larger dimensions (or the empty set). We ignore these kinds of intersections.} Clearly $C$ is the convex hull of $\{p_1,\dots,p_k\}$. 
\end{proof}

\begin{theorem}\label{thm: PS captures polytopes}
For every polytope $\Pvec$, there is a formula $\varphi\in \Psm$ such that $\Pvec_{\varphi}=\Pvec$.
\end{theorem}

\begin{proof}
By Lemma \ref{lemma: polytopes as convex hulls}, $\Pvec$ is the convex hull of a finite number $k$ of points. Each such point is definable in $\Delta^{n-1}$ by $2n-2$ monic inequalities -- say, the point $p_j$ of coordinates $(a^j_1,\dots,a^j_{n-1},1-(a^j_1,\dots,a^j_{n-1}))$ is defined by the equalities $\epsilon_1 = a^j_1,\dots, \epsilon_{n-1} = a^j_{n-1}$, which are equivalent to pairs of inequalities. By Theorem \ref{thm: characterization of P minus}, for each of these inequalities there is a formula $\psi^j_i$ in $\PP^- \subseteq \Psm$ such that $\Pvec_{\psi_i^j}$ is the subset of $\Delta^{n-1}$ defined by the corresponding inequality. 
Thus $\Pvec$ is the convex hull of the sets $\Pvec_{\bigwedge_{i}\psi^j_i}$, for $j=1,\dots,k$. Noting now that the convex hull of $k$ points can be obtained iteratively by taking the convex hull of two points $p_1,p_2$, then the convex hull of the resulting set and $p_3$, and so on, we can apply $k-1$ times 
Lemma \ref{lemma: strict tensor as convex hull} to prove that the formula $\varphi=\bigvee_{j=1..k}\bigwedge_{i = 1..2n-2}\psi^j_i$ is such that $\Pvec_\varphi= \Pvec$.
\end{proof}

\begin{corollary}\label{cor: Ps separation}
    $\PO < \Ps$ (and thus $\PP < \Ps$). 
\end{corollary}

 \begin{proof}
        As already mentioned, $\PO \leq \Ps$, since every formula of the form $\alpha\supset\psi$ is equivalent to $\alpha^d\lor\psi$ (where $\alpha^d$ is given by recursive clauses without occurrences of $\supset$).

    Now let us prove the strict inclusion; for this refer to Figure \ref{figure: Ps undefinable in PO}. Let us work in the tridimensional space, and consider the linear inequality $2\epsilon_1 + \epsilon_2 \vartriangleleft 1$ and the equation of its surface, $2\epsilon_1 + \epsilon_2 = 1$. Its intersection with $\Delta^2$ is a line $N$ with the following properties: 1) its projection on the $(\epsilon_1,\epsilon_2)$ plane (the thick line in figure 5) does not pass through the origin, and 2) it is not parallel to any side of $\Delta^2$, since it intersects each of them.

    Note that $N$, being a polytope, is the probability set of a formula $\varphi$ of $\PP(\lor)$ (by Theorem \ref{thm: PS captures polytopes}). We show that it is not a 
    signed binary probability set; thus, by Theorem \ref{thm: characterization of PO}, it is not the probability set of a $\PO$ formula, and then $\varphi$ cannot be equivalent to any $\PO$ formula.

    To this effect, we first show that $N$ cannot be included, as a subset, in the perimeter of any signed binary probability set that is defined by a single inequality. Such a set is defined either by a monic inequality or by a signed binary inequality with constant coefficient $0$. In the former case, $N$ intersects the perimeter (a line) of such space in at most one point, by 2) (remember from the proof of lemma \ref{lemma: not all signed monic are monic} that lines defined by a monic equality are parallel to the sides of the simplex). In the latter case, $N$ intersects such a line in at most one point by 1) (since all lines defined by a signed binary equality of constant coefficient $0$ pass through the origin when projected to the $(\epsilon_1,\epsilon_2)$ plane). Thus, $N$ is not included in the boundary of such probability spaces.

    The proof is then extended to general signed binary probability spaces as in Lemma \ref{lemma: not all signed monic are monic}.
    \end{proof}







\begin{figure}
    \centering
    \begin{tikzpicture}[scale=2.5]
        \draw[->] (-0.2, 0) -- (1.2, 0) node[right] {$\epsilon_1$};
        \draw[->] (0, -0.2) -- (0, 1.2) node[above] {$\epsilon_2$};

        \filldraw[fill=gray!20, draw=black] (0, 0) -- (0, 1) -- (1, 0) -- cycle;


        \draw[line width = 0.7mm] (0.5,0) -- (0,1);
        
        \draw[dashed] (0,0) -- (0.25,0.75);
        \draw[dashed] (0.3,0) -- (0.3,0.7);
        \draw[dashed] (0,0.1) -- (0.9,0.1);
        \draw[dashed] (0.8,0) -- (0,0.8);

        \node[below left] at (0, 0) {(0, 0)};
        \node[above] at (-0.22, 0.9) {(0, 1)};
        \node[below right] at (0.8, 0) {(1, 0)};
    \end{tikzpicture}
    \caption{Projection of the standard simplex in the $(\epsilon_1,\epsilon_2)$-plane. The thick line is the frontier of the object defined by the inequality $2\epsilon_1+\epsilon_2\leq 1$. The dashed lines exemplify frontiers of $\PO$-definable sets.}
    \label{figure: Ps undefinable in PO}
\end{figure}


We remark that definability of polytopes does not exhaust the expressivity of language $\Psm$, since there are already $\PP^-$ formulas that capture the probability set of inequalities such as $\epsilon_1+\epsilon_2 > \frac{1}{2}$, that are clearly not (finite unions of) polytopes.

Let us call a \textbf{semipolytope} any subset of $\Delta^{n-1}$ defined by a finite system of linear inequalities in $\epsilon_1,\dots,\epsilon_n$.

\begin{theorem}\label{thm: PS captures semipolytopes}
For every semipolytope $\Pvec$, there is a formula $\varphi\in \Psm$ such that $\Pvec_{\varphi}=\Pvec$.
\end{theorem}

\begin{proof}
We proceed by induction on the dimension of the semipolytope. Figure \ref{figure: defining semipolytopes in PCOs} illustrates our construction of the formula $\varphi$ from a 2-dimensional semipolytope. 

Observe first that any single point is a polytope, thus definable in $\Psm$. The same goes for any segment that includes the extremities (since it is the part of a line delimited by two parallel hyperplanes). A segment without one or both extremities is not a polytope, but it is still definable in $\Psm$; the extremity can be removed by taking the intersection of the segment with a (monic) open half-space. Thus, all semipolytopes of  dimension $\leq 1$ are definable in $\Psm$.

Now suppose we have proved the claim for semipolytopes of dimension $\leq n$, and let $\Pvec$ be a semipolytope of dimension $n+1$. Its frontier consists of a finite number of facets $F_1,\dots,F_k\subset \Pvec$ plus possibly other facets not included in $\Pvec$. The facets $F_1,\dots,F_k\subset \Pvec$ are sets of dimension $\leq n$, and we can assume wlog that they do not contain any point of their own, lower-dimensional frontier (so they can consist of a single point, a segment without the two extremities, a polygon without the perimeter, etc.); this guarantees that they are semipolytopes. Let us also enumerate the vertices of $\Pvec$ as $v_1,\dots,v_s$ (which may or may not belong to $\Pvec$).

Fix a point $a$ in the interior of $\Pvec$. For each vertex $v_m$, let $l_m$ be the segment $[a,v_m)$. $l_m\subseteq \Pvec$, since $\Pvec$ is a convex set. By the base cases, $l_m$ is definable by a $\Psm$ formula $\psi_m$. Then, since, by 
Lemma \ref{lemma: strict tensor as convex hull}, applying $\lor$ to two formulas gives a formula whose probability set is the convex hull of the probability sets of the two formulas, the interior of $\Pvec$ (call it $I$) is defined by the $\Psm$ formula $\psi := \psi_1 \lor \dots \lor \psi_s$. Now, by the induction hypothesis, each of the $F_i$ is definable by a $\Psm$ formula $\chi_i$. Thus, finally, $\Pvec$ is definable by the $\Psm$ formula $\psi \sqcup \chi_1\sqcup\dots\sqcup \chi_k$. \qed
\end{proof}

\begin{figure}
$$\begin{tikzpicture}[scale=0.75]
\node (x) at (3,3) {How to capture};
\draw[very thick,dashed,shift={(6,3)}] (.1,1)--(-1,.2)--(-.5,-.9)--(.7,-.9)--(1.2,.3)--cycle;
\fill [gray, opacity=.4,shift={(6,3)}] (.1,1)--(-1,.2)--(-.5,-.9)--(.7,-.9)--(1.2,.3)--cycle;
\draw[very thick,shift={(6,3)}] (.1,1)--(-1,.2);
\draw[very thick,shift={(6,3)}] (-.5,-.9)--(.7,-.9);
\fill[shift={(6,3)}] (.1,1)  circle (0.08);
\fill[shift={(6,3)}] (-1,.2)  circle (0.08);
\fill[shift={(6,3)}] (-.5,-.9) circle (0.08);
\fill[shift={(6,3)}] (.7,-.9)  circle (0.08);
\fill[shift={(6,3)}] (1.2,.3)  circle (0.08);
\fill[white,shift={(6,3)}] (.1,1)  circle (0.06);
\fill[white,shift={(6,3)}] (1.2,.3)  circle (0.06);
\node (x) at (8.8,3) {in $\mathcal{PCO}(\lor)$?};
\draw[very thick,dashed] (.1,1)--(-1,.2)--(-.5,-.9)--(.7,-.9)--(1.2,.3)--cycle;
\fill (0,0)  circle (0.08);
\node (x) at (.2,0) {\small $a$};
\node (x) at (0,-1.5) {\small Pick a point $a$};
\node (x) at (0,-2) {\small in the interior.};
\node (x) at (2,0) {$\mapsto$};
\draw[very thick,dashed,shift={(4,0)}] (.1,1)--(-1,.2)--(-.5,-.9)--(.7,-.9)--(1.2,.3)--cycle;
\fill (4,0)  circle (0.08);
\draw[very thick,shift={(4,0)}] (0,0)--(.1,1);
\draw[very thick,shift={(4,0)}] (0,0)--(-1,.2) ;
\draw[very thick,shift={(4,0)}] (0,0)--(-.5,-.9);
\draw[very thick,shift={(4,0)}] (0,0)--(.7,-.9);
\draw[very thick,shift={(4,0)}] (0,0)--(1.2,.3);
\fill[shift={(4,0)}] (.1,1)  circle (0.08);
\fill[shift={(4,0)}] (-1,.2)  circle (0.08);
\fill[shift={(4,0)}] (-.5,-.9) circle (0.08);
\fill[shift={(4,0)}] (.7,-.9)  circle (0.08);
\fill[shift={(4,0)}] (1.2,.3)  circle (0.08);
\fill[white,shift={(4,0)}] (.1,1)  circle (0.06);
\fill[white,shift={(4,0)}] (-1,.2)  circle (0.06);
\fill[white,shift={(4,0)}] (-.5,-.9) circle (0.06);
\fill[white,shift={(4,0)}] (.7,-.9)  circle (0.06);
\fill[white,shift={(4,0)}] (1.2,.3)  circle (0.06);
\node (x) at (4.3,.5) {\small $\psi_1$};
\node (x) at (4.7,-0.1) {\small $\psi_2$};
\node (x) at (4.2,-.6) {\small $\psi_3$};
\node (x) at (3.53,-.3) {\small $\psi_4$};
\node (x) at (3.6,.32) {\small $\psi_5$};
\node (x) at (4,-1.5) {\small Draw half-lines};
\node (x) at (4,-2) {\small from $a$ to vertices.};
\node (x) at (4,-2.5) {\small By induction, defined};
\node (x) at (4,-3) {\small by formulas $\psi_1,\dots,\psi_5$.};
\node (x) at (6,0) {$\mapsto$};
\draw[very thick,dashed,shift={(8,0)}] (.1,1)--(-1,.2)--(-.5,-.9)--(.7,-.9)--(1.2,.3)--cycle;
\fill (8,0)  circle (0.08);
\fill [gray, opacity=.4,shift={(8,0)}] (.1,1)--(-1,.2)--(-.5,-.9)--(.7,-.9)--(1.2,.3)--cycle;
\draw[very thick,shift={(8,0)}] (0,0)--(.1,1);
\draw[very thick,shift={(8,0)}] (0,0)--(-1,.2) ;
\draw[very thick,shift={(8,0)}] (0,0)--(-.5,-.9);
\draw[very thick,shift={(8,0)}] (0,0)--(.7,-.9);
\draw[very thick,shift={(8,0)}] (0,0)--(1.2,.3);
\fill[shift={(8,0)}] (.1,1)  circle (0.08);
\fill[shift={(8,0)}] (-1,.2)  circle (0.08);
\fill[shift={(8,0)}] (-.5,-.9) circle (0.08);
\fill[shift={(8,0)}] (.7,-.9)  circle (0.08);
\fill[shift={(8,0)}] (1.2,.3)  circle (0.08);
\fill[white,shift={(8,0)}] (.1,1)  circle (0.06);
\fill[white,shift={(8,0)}] (-1,.2)  circle (0.06);
\fill[white,shift={(8,0)}] (-.5,-.9) circle (0.06);
\fill[white,shift={(8,0)}] (.7,-.9)  circle (0.06);
\fill[white,shift={(8,0)}] (1.2,.3)  circle (0.06);
\node (x) at (8,-1.5) {\small Take convex hull};
\node (x) at (8,-2) {\small of half-lines.};
\node (x) at (8,-2.5) {\small Defined by};
\node (x) at (8,-3) {\small $\chi:\psi_1\lor\dots\lor\psi_5$.};
\node (x) at (10,0) {$\mapsto$};
\draw[very thick,dashed,shift={(12,0)}] (.1,1)--(-1,.2)--(-.5,-.9)--(.7,-.9)--(1.2,.3)--cycle;
\fill [gray, opacity=.4,shift={(12,0)}] (.1,1)--(-1,.2)--(-.5,-.9)--(.7,-.9)--(1.2,.3)--cycle;
\draw[very thick,shift={(12,0)}] (.1,1)--(-1,.2);
\draw[very thick,shift={(12,0)}] (-.5,-.9)--(.7,-.9);
\fill[shift={(12,0)}] (.1,1)  circle (0.08);
\fill[shift={(12,0)}] (-1,.2)  circle (0.08);
\fill[shift={(12,0)}] (-.5,-.9) circle (0.08);
\fill[shift={(12,0)}] (.7,-.9)  circle (0.08);
\fill[shift={(12,0)}] (1.2,.3)  circle (0.08);
\fill[white,shift={(12,0)}] (.1,1)  circle (0.06);
\fill[white,shift={(12,0)}] (1.2,.3)  circle (0.06);
\node (x) at (10.73,.3) {\small $\theta_1$};
\node (x) at (11.4,.8) {\small $\theta_2$};
\node (x) at (12,-1.2) {\small $\theta_3$};
\node (x) at (12,-1.8) {\small Add missing facets};
\node (x) at (12,-2.3) {\small (which by induction are};
\node (x) at (12,-2.8) {\small defined by $\theta_1,\theta_2,\theta_3$)};
\node (x) at (12,-3.3) {\small by taking the formula};
\node (x) at (12,-3.8) {\small $\chi\sqcup\theta_1\sqcup\theta_2\sqcup\theta_3$.};
\end{tikzpicture}$$
\caption{How to define a 2-dimensional semipolytope in $\PCOs$.}\label{figure: defining semipolytopes in PCOs}
\end{figure}


\begin{corollary}\label{cor: PS captures linear sets}
 Any linear subset of $\Delta^{n-1}$ is the probability set of a $\Psm$ formula.   
\end{corollary}

\begin{proof}
Any such set $S$ is a finite union of semipolytopes $S_1,\dots,S_k$. By theorem \ref{thm: PS captures semipolytopes}, each $S_i$ is the probability set of some formula $\psi_i\in\Psm$.   Thus $S$ is the probability set of the $\Psm$ formula $\psi_1\sqcup\dots\sqcup\psi_k$. 
\end{proof}

For any signature $\sigma$, $\K^\sigma_\bot$ denotes the set of all empty causal multiteams of signature $\sigma$. Write $\mathbb C_\sigma$ for the set of all causal multiteams of signature $\sigma$.

\begin{corollary}
    The language $\Psm$ is closed under weak contradictory negation; i.e., whenever a set $\K$ of causal multiteams of signature $\sigma$ is definable in $\Psm_\sigma$, also  $(\mathbb C_\sigma \setminus \K) \cup \K^\sigma_\bot$ is.
\end{corollary}

\begin{proof}
    Suppose $\K$ is defined by $\varphi\in\Psm_\sigma$. Then, by Corollary \ref{cor: PS is linear}, $\Pvec_\varphi$ is linear, i.e. the union of finitely many semipolytopes $S_1\cup\dots\cup S_k$. Its complement is $\Delta^{n-1} \setminus (S_1\cup\dots\cup S_k) = (\Delta^{n-1} \setminus S_1) \cap \dots \cap (\Delta^{n-1} \setminus S_k)$.

    Now, it is not difficult to see that the complement of a semipolytope is a finite union of semipolytopes\footnote{If $S$ is a semipolytope defined by a system of inequalities $e_1,\dots,e_l$, let $\overline e_i$ ($i=1\dots l$) be the ``complementary'' inequality obtained replacing $\geq$ with $<$, $\leq$ with $>$ and vice versa. Then $\Delta^{n-1} \setminus S$ is the union of the semipolytopes defined by $\overline e_1,\dots,\overline e_l$, respectively.}; say, $\Delta^{n-1} \setminus S_i = A_1^1 \cup\dots\cup A_n^1$, where we can assume wlog that the number $n$ of these sets is the same for each $i$.  Thus, $\Delta^{n-1} \setminus (S_1\cup\dots\cup S_k) = \bigcup_{\pi: \{1,\dots, n\}\rightarrow\{1,\dots, n\}}(A^1_{\pi(1)}\cap \dots\cap A^n_{\pi(n)})$. Since it can be easily proved (along the lines of Lemma \ref{lemma: preserve polytopes}) that the intersection of semipolytopes is a semipolytope, we conclude that $\Delta^{n-1} \setminus \Pvec_\varphi$ itself is a finite union of semipolytopes, i.e. a linear set. Thus, by Corollary \ref{cor: PS captures linear sets}, there is a formula $\theta\in \Psm$ such that $\Pvec_\theta = \Delta^{n-1} \setminus \Pvec_\varphi$. Since, by the empty team property and the closure of $\Psm$ under rescalings and change of causal laws, 
    both $\K$ and $\K_\theta$ are maximal among sets of causal multiteams of signature $\sigma$  having probability set $\Pvec_\varphi$, resp. $\Pvec_\theta$, we conclude that $\K\cup\K_\theta = \mathbb C_\sigma$, and that $\K^\sigma_\bot \subseteq \K_\theta,\K$. Thus $\K_\theta= (\mathbb C_\sigma \setminus \K)\cup \K^\sigma_\bot$. \qed 
\end{proof}

\noindent We do not know, however, whether the weak contradictory negation is definable in any syntactic sense.

\begin{theorem}\label{thm: characterization of Ps}
A class $\K$ of multiteams of signature $\sigma$ is definable by a formula of $\Psm$ (resp. $\Ps$) if and only if: 
\begin{enumerate}
\item $\K$ contains all empty causal multiteams of signature $\sigma$
\item $\K$ is closed under change of laws
\item $\K$ is closed under rescaling
\item $\K$ is linear.
\end{enumerate}
\end{theorem}



\begin{proof}
We first prove the statement for $\Psm$.

$\Rightarrow$) Suppose that $\K$  is definable by a formula $\varphi$ of $\Psm$. 
By Lemma \ref{lemma: rescaling for PCOs}, $\K$ is closed under rescaling.    
By lemma \ref{lemma:trivial laws}, $\K$ is closed under change of laws. By Corollary  \ref{cor: PS is linear}, $\K$ is linear. Finally, by Corollary \ref{cor: PCSOs in PCOomega} $\K$ is definable in $\PCO^\omega$, and so by Theorem \ref{theorem: expressivity of PCOinf} it contains all empty causal multiteams of signature $\sigma$.

$\Leftarrow$) Since $\K$ is linear, Corollary \ref{cor: PS captures linear sets} guarantees that $\Pvec_\K$ is the probability set of some $\Psm$ formula $\varphi$. The other three conditions ensure that $\K$ is the largest set of causal multiteams of signature $\sigma$ with probability set $\Pvec_\K$. The analogous three properties of $\varphi$ guarantee that also $\K_\varphi$ is the largest such set. Thus $\K=\K_\varphi$.

For the case of $\Ps$, the $\Rightarrow$ direction is proved in the same way (remember that Corollary \ref{cor: PCSOs in PCOomega} and the relevant lemmas apply to $\Ps$ formulas as well). The $\Leftarrow$ direction holds a fortiori given the result for $\Psm$. \qed
%
%
\end{proof}

\begin{corollary}
    $\Psm \equiv \Ps$; in particular, $\Ps$ is closed under weak contradictory negation.
\end{corollary}

We can now move towards a characterization of language $\PCOs$.  



\begin{lemma}\label{lemma: cf distributes over strict tensor}
  The operator $\cf$ distributes over the strict tensor $\lor$, i.e., for any $\psi,\chi\in\PCOs$, $\SET X = \SET x \cf (\psi\lor\chi)$ is equivalent to $(\SET X = \SET x \cf \psi) \lor (\SET X = \SET x \cf \chi)$.  
\end{lemma}

\begin{proof}
    $\Leftarrow$) Suppose $T\models (\SET X = \SET x \cf \psi) \lor (\SET X = \SET x \cf \chi)$.
Then there are submultiteams $T_1 = (T_1^-,\F)$ and $T_2 = (T_2^-,\F)$ of $T$, with $T_1^-\cup T_1^-=T^-$ and $T_1^-\cap T_2^- = \emptyset$, such that $T_1\models \SET X = \SET x \cf\psi$ and $T_2\models \SET X = \SET x \cf\chi$. Thus, $(T_1)_{\SET X = \SET x}\models \psi$ and $(T_2)_{\SET X = \SET x}\models \chi$. Now, since $T_1^-$ and $T_2^-$ are disjoint subsets of $T^-$, the assignments in $T_1^-$ have different $Key$ values than those in $T_2^-$. Since the $Key$ values are not modified by interventions, then, the same holds for $(T_1)_{\SET X = \SET x}^-$ and $(T_2)_{\SET X = \SET x}^-$; thus, $(T_1)_{\SET X = \SET x}^-$ and $(T_1)_{\SET X = \SET x}^-$ are disjoint. Furthermore, it is easy to see that $(T_1)_{\SET X = \SET x}^-\cup(T_1)_{\SET X = \SET x}^- = T_{\SET X = \SET x}^-$. Then $T_{\SET X = \SET x}\models \psi\lor\chi$, i.e. $T\models \SET X = \SET x \cf (\psi\lor\chi)$. 

    $\Rightarrow$) Suppose  $T\models \SET X = \SET x \cf (\psi\lor\chi)$. Then there are $S_1=(S_1^-,\F),S_2=(S_2^-,\F) \leq T_{\SET X = \SET x}$ with $S_1^-\cup S_2^- = T_{\SET X = \SET x}^-$, $S_1^-\cap S_2^- = \emptyset$, $S_1\models \psi$ and $S_2\models \chi$. Now, let $\pi$ be the bijection between $T_{\SET X = \SET x}^-$ and $T^-$ that sends a $t\in T_{\SET X = \SET x}^-$ to the unique $s\in T^-$ with the same value for $Key$. For $i=1,2$, let $T_i^- = \pi[S_i^-]$, and $T_i = (T_i^-,\F)$. Clearly, then, $S_i= (T_i)_{\SET X = \SET x}$. Thus, $T_1\models \SET X = \SET x \cf \psi$ and $T_2 \models \SET X = \SET x \cf \chi$. Furthermore, by construction, $T_1^-\cup T_2^- = T^-$ and $T_1^-\cap T_2^- = \emptyset$.  
\qed    
\end{proof}





\begin{lemma}\label{lemma: ignoring cf in PCOS}
Let $\varphi\in \PCOs_{\sigma}$,
 and $\F$ a function component of signature $\sigma$. Then there is a formula $\varphi^\F\in \Ps_{\sigma}$ 
  such that, for every causal multiteam $T$ of signature $\sigma$ and function component $\F$,
\[
T\models \varphi \iff T \models \varphi^\F.
\]
\end{lemma} 

\begin{proof}
By the same proof as in Proposition \ref{proposition: ignoring cf}, using also Lemma \ref{lemma: cf distributes over strict tensor} to show that we may assume that $\cf$ only occurs in the arguments of $\Pr$ within $\PCOs$ formulas. \qed
\end{proof}

\begin{theorem}\label{thm: characterization of PCOs}
Let $\K$ be a class of causal multiteams of signature $\sigma$. $\K$ is definable by a $\PCOs_{\sigma}$ formula if and only if:
\begin{enumerate}
\item it has the empty multiteam property 
\item it is closed under rescaling 
\item $\K = \bigcup_{\F\in\mathbb{F}_\sigma}\K^\F$, where  each $\K^\F$ is a linear set of causal multiteams of team component $\F$.
\end{enumerate}
\end{theorem}

\begin{proof}
    The argument is very similar as in the proof of Theorem \ref{thm: characterization of PC}, using Theorem \ref{thm: characterization of Ps} instead of Theorem \ref{thm: characterization of P}.
\end{proof}

\begin{proposition}\label{prop: Ps strictly included in PCOs}
   $\PCOs < \PCO^\omega$. 
\end{proposition}

\begin{proof}
     The strict inclusion $\PCOs < \PCO^\omega$ follows either from a counting argument or from the fact that the probability sets of $\PCOs$ formulas are linear 
     while the probability sets of $\PCO^\omega$ formulas have no restrictions. \qed
\end{proof}

\begin{proposition}\label{prop: PCO Ps incomparable}
1) $\PCO$ and $\Ps$ are incomparable. 

2) $\PC$ and $\Ps$ are incomparable.
\end{proposition}

\begin{proof}
By the inclusion of $\PC$ in $\PCO$, it suffices to prove that  $\PC$ is not included in $\Ps$, and that $\Ps$ is not included in $\PCO$.

$\PC$ is not included in $\Ps$. This follows from the already mentioned fact that $\PC$ is not closed under change of laws, while $\Ps$ is (lemma \ref{lemma:trivial laws}).

$\Ps$ is not included in $\PCO$. By Corollary \ref{cor: Ps separation}, there is a formula $\varphi\in \Ps$ that is not signed binary. By Lemma \ref{lemma:trivial laws}, $\varphi$ is closed under change of causal laws. Thus, for any $\F$, the probability set of $\K_\varphi^\F$ is equal to the probability set of $\K_\varphi$, and is thus not signed binary. Thus, by Theorem \ref{thm: characterization of PCO}, there is no $\PCO$ formula equivalent to $\varphi$. \qed
\end{proof}

\begin{proposition}\label{prop: PCO strictly included in PCOs}
   $\PCO < \PCOs$. 
\end{proposition}

\begin{proof}
By Proposition \ref{prop: PCO Ps incomparable}, $\Ps \not\leq \PCO$, and thus a fortiori $\PCOs \not\leq \PCO$. \qed    
\end{proof}

\section{Definability of probabilistic and dependence atoms}\label{sec:teamsemantics}

Next we briefly explore the relationships between our logics and the probabilistic atoms studied in probabilistic and multiteam semantics.
We consider also the dependence atom by V\"a\"an\"anen \cite{Vaa2007}, and marginal distribution identity and probabilistic independence atoms by Durand et al. \cite{DurHanKonMeiVir2016}.

The dependence atom $\dep{\SET X}{\SET Y}$ expresses that the values of $\SET X$ functionally determine the values of $\SET Y$.
Dependence atoms can be expressed already in $\PO$:
\[
\dep{\SET X}{\SET Y} \dfn  \bigwedge_{\SET{x}\in \ran(\SET{X})} \bigsqcup_{\SET{y}\in \ran(\SET{Y})} \SET X = \SET x \supset \SET Y= \SET y
\]

The marginal distribution identity atom $\SET{X} \approx \SET{Y}$ states that the marginal distributions induced by $\SET{X}$ and $\SET{Y}$ are identical. This  can be defined in $\PP$ by
\begin{align*}
\SET{X} \approx \SET{Y} \dfn   &    \bigwedge_{\SET{x}\in \ran(\SET{X})\cap \ran(\SET{Y})} \Pr(\SET X=\SET x)= \Pr(\SET Y=\SET x) \land \\
  &    \bigwedge_{\SET{x}\in \ran(\SET{X})\setminus \ran(\SET{Y})} \Pr(\SET X=\SET x)=0 \land \bigwedge_{\SET{y}\in \ran(\SET{Y})\setminus \ran(\SET{X})} \Pr(\SET Y=\SET y)=0.
\end{align*}
The conditional probabilistic atoms inherit their semantics from probability theory: 
%
\begin{align*}
&T\models \Pr(\alpha\mid\beta) \vartriangleright \epsilon  &&\text{iff}&& (T^\beta)^- = \emptyset \text{ or } P_{T^\beta}(\alpha) \vartriangleright \epsilon.\\
&T\models \Pr(\alpha\mid\beta)\vartriangleright \Pr(\gamma\mid\delta) &&\text{iff}&& (T^\beta)^- = \emptyset \text{ or } (T^\delta)^- = \emptyset \text{ or } P_{T^\beta}(\alpha) \vartriangleright P_{T^\delta}(\beta),
\end{align*}
and we may also write e.g. $\Pr(\alpha\mid\beta)\vartriangleright \Pr(\gamma)$ as an abbreviation for $\Pr(\alpha\mid\beta)\vartriangleright \Pr(\gamma\mid\top)$. 
 Related to these, the atom $\SET{X} \pindep_{\SET{Z}} \SET{Y}$ (\emph{conditional independence atom}) states that for any given value for the variables in $\SET{Z}$ the variable sets $\SET{X}$ and $\SET{Y}$ are probabilistically independent. Its special case with $\SET Z=\emptyset$ is called \emph{marginal independence atom}. We can define these atoms in terms of conditional comparison atoms:
\begin{align*}
\SET{X} \pindep \SET{Y} &\dfn \bigwedge_{\substack{\SET{x}\in \ran(\SET{X})\\ \SET{y}\in \ran(\SET{Y})}} \Pr(\SET X=\SET x)= \Pr(\SET X=\SET x \mid \SET Y= \SET y)\\
\SET{X} \pindep_{\SET{Z}} \SET{Y} &\dfn \bigwedge_{\substack{\SET{x}\in \ran(\SET{X})\\ \SET{y}\in \ran(\SET{Y})\\\SET{z}\in \ran(\SET{Z})}} \Pr(\SET X=\SET x \mid \SET Z=\SET z)= \Pr(\SET X=\SET x \mid \SET Y \SET Z= \SET y \SET z)
\end{align*}
Hence the atoms (and the dependence atom expressed as $\SET{Y} \pindep_{\SET{X}} \SET{Y}$) are expressible in $\PP$ extended with the conditional probability comparison atoms. 

The above definitions of atoms imply that our languages, if enriched with conditional probability atoms and arbitrary applications of the disjunction $\lor$, are strong enough to the express properties of multiteams that are expressible in the quantifier free fragments of the logics $\mathrm{FO}(\pindep)$ (\emph{probabilistic independence logic}) and $\mathrm{FO}(\approx)$ (\emph{probabilistic inclusion logic}), over any fixed finite structure. The expressivity and complexity of these logics have been thoroughly studied in the probabilistic and multiteam semantics literature (see \cite{DurHanKonMeiVir2016,DurHanKonMeiVir2018,GradelW2022, HanHirKonKulVir2019,HKBV2020,HanVir2022,Wilke2022}).

It was observed in \cite{BarSan2024} that $\Pr(\alpha\mid\gamma) \vartriangleright \epsilon$ and $\Pr(\alpha\mid\gamma)\vartriangleright \Pr(\beta\mid\gamma)$ can be defined by $\gamma \supset\Pr(\alpha) \vartriangleright \epsilon$ and $\gamma \supset\Pr(\alpha) \vartriangleright \Pr(\beta)$, respectively.
The latter result concerns comparison atoms in which both probabilities are conditioned over the same formula, $\gamma$.
We establish that this restriction is necessary, and that  $\Pr(\alpha\mid\gamma)\geq \Pr(\beta\mid\delta)$ is not, in general, expressible in $\PCOs$. As a consequence, the definability status of (conditional) independence atoms in $\PCOs$ remains unsettled. 
\begin{restatable}{theorem}{undefcomparison}\label{thm: undefinability of conditional comparison}
The comparison atoms
$
\Pr(\alpha \mid \beta)\vartriangleleft \Pr(\gamma \mid \delta)$ and 
$\Pr(\alpha \mid \beta)\vartriangleleft \Pr(\gamma)$,
(where $\vartriangleleft \in \{\leq,\geq,<,>,=\}$) are not, in general, expressible in $\PCOs$.
\end{restatable}

\begin{proof}
    Due to the equivalence $\Pr(\alpha \mid \beta)\vartriangleleft \Pr(\gamma \mid \top) \equiv \Pr(\alpha \mid \beta)\vartriangleleft \Pr(\gamma)$, it suffices to prove the theorem for $\Pr(\alpha \mid \beta)\vartriangleleft \Pr(\gamma)$. 

    We fix a signature $\sigma$ that allows exactly $4$ distinct assignments $s_i,s_j,s_k,s_l\in \mathbb  B_\sigma$ (we will then be working in a $4$-dimensional vector space), and a constant $\delta\in (0,1]\cap \mathbb Q$. The proof proceeds by showing that the conjunction
    \[
    \Xi \dfn \Pr(\Al_k \lor \Al_i \mid \Al_l \lor \Al_i) \vartriangleleft  \Pr(\Al_l \lor \Al_j) \ \land \  \Pr(\Al_i)= \delta  \ \land \   \Pr(\Al_j \lor \Al_k \lor \Al_l) = 1-\delta
    \]
    has a probability set that cannot be characterized in terms of systems of linear inequalities, and thus is not expressible in $\Ps$; extending the result to the whole $\PCO(\lor)$ is then straightforward. 

    Let $T$ be a causal multiteam satisfying $T\models \Xi$. For ease of reading, we will write $e_i,e_j,e_k,e_l$ for $\epsilon_{s_i}^T,\epsilon_{s_j}^T,\epsilon_{s_k}^T,\epsilon_{s_l}^T$, respectively. 
    Since $T\models \Pr(\Al_i)= \delta $ and $\delta >0$, we have in particular that $e_l + e_i > 0$. Thus, the following equivalences hold:
    \begin{align*}
    & T\models \Pr(\Al_k \lor \Al_i \mid \Al_l \lor \Al_i) \vartriangleleft  \Pr(\Al_l \lor \Al_j) \\
    \iff & P_T((\Al_k \lor \Al_i)\land(\Al_l \lor \Al_i))  \vartriangleleft P_T(\Al_l \lor \Al_j) \cdot P_T(\Al_l \lor \Al_i) \\
    \iff & P_T(\Al_i)  \vartriangleleft P_T(\Al_l \lor \Al_j) \cdot P_T(\Al_l \lor \Al_i) \\
    \iff & e_i \vartriangleleft (e_l + e_j)\cdot (e_l + e_i)\\
    \iff & e_i \vartriangleleft (e_l)^2 +  e_l e_i + e_j e_l + e_j e_i.
    \end{align*}
    Thus, the boundary of the probability set $\Pvec_{\Xi}$ 
    consists of those quadruples $(e_i, e_j,e_k.e_l)$ that belong to the curve (in 4-dimensional space) $C$ satisfying the following system: 
       \[
\left\{
\begin{array}{l}
  \delta = \epsilon_l^2 +  \delta \epsilon_l +  \epsilon_j \epsilon_l + \delta \epsilon_j    \\
  \epsilon_i=\delta    \\  
  \epsilon_j + \epsilon_k + \epsilon_l = 1-\delta  
\end{array}
\right.
\]
where we have replaced $\epsilon_i$ in the first equality using the second equality.

 The first equation defines a ``cylinder'' $CL$ whose sections by planes parallel to the $(\epsilon_j, \epsilon_l)$ plane take the form of the conic of equation $\delta = \epsilon_l^2 +  \delta \epsilon_l +  \epsilon_j \epsilon_l + \delta \epsilon_j$, call it $D$; note that this curve is also the projection of $CL$ on the $(\epsilon_j, \epsilon_l)$ plane. 
  We observe  that, if $D$ is a nondegenerate conic, then also $C$ is; note indeed that no nondegenerate conic can be a projection of a union of lines and points -- in particular, not a projection of a degenerate conic. Let us then prove, by standard methods, that $D$ is nondegenerate.
 %
  %
  %
%
 %
   The equation of $D$  can also be rewritten as  $2 \epsilon_{l}^2 +  2\delta \epsilon_{l} + 2 \epsilon_{j} \epsilon_{l} + 2\delta \epsilon_{j} - 2\delta = 0$.
    It is degenerate iff its discriminant is $0$ (see e.g. [33]). The discriminant is
\begin{equation*}
 \operatorname{det}\left(\begin{array}{ccc}
 2 & 1 & \delta\\
 1 & 0 & \delta\\
 \delta & \delta & 2\delta \\
 \end{array}\right)    = -( 2\delta - \delta^2) - \delta(2\delta-\delta) = -2\delta  
 \\
\end{equation*}
which is $0$ if and only if $\delta$ is $0$. Thus, with our choice of $\delta >0$, the conic $D$ is nondegenerate. By the remarks, above, we can conclude that also $C$ is a nondegenerate conic; thus, it cannot be contained in the boundary $\Pvec_\Xi$ if $\Xi$ is a linear probability set. Thus, we may conclude that $\Xi$ is not expressible in $\Ps$, and neither is $\varphi$.


Let us then consider the case of $\PCOs$. Recall that there are only a finite number of $\F\in\mathbb{F_\sigma}$. So, since a nondegenerate conic is an infinite set, 
there is an $\F$ such that the $(\epsilon_j,\epsilon_l)$ projection of the boundary of the probability set of $\Xi \land\Phi^\F$ contains an (infinite) segment of the nondegenerate conic $C$. Hence, as above, we conclude that the probability set of $\Xi \land\Phi^\F$ is not defined by linear inequalities. Thus, by  Theorem \ref{thm: characterization of PCOs},  $\Xi \land\Phi^\F$ is not expressible $\PCOs$, and neither is $\varphi$. 
\end{proof}

\section{Conclusion}\label{sec:conclusion}

We embarked for a comprehensive study of the expressive power of logics of probabilistic reasoning and causal inference in the unified setting of causal multiteam semantics.
We focused on the logic $\PCO$ that can express probability comparisons in a dataset, and encompasses interventionist counterfactuals and selective implications for describing consequences of actions and consequences of learning from observations, respectively. In addition, we considered the syntactic fragments $\PP^-$, $\PP$, $\PO$, and $\PC$ of $\PCO$ and proved that they form a strict expressivity hierarchy (see Figure \ref{fig:FIRSTCAUSALGRAPH} on page \pageref{fig:FIRSTCAUSALGRAPH}). We showed that this hierarchy can be further extended by considering languages $\Psm,\Ps$ and $\PCOs$ that extend $\PP^-, \PP$ and $\PCO$, respectively, with the strict tensor operator typical of the literature on multiteam semantics. For each of these languages, we discovered natural complete characterizations, based on the families of linear equations needed to define the corresponding classes of causal multiteams (satisfying some invariances); these results are summarized in Table \ref{table:exp} (on page \pageref{table:exp}). Interestingly, while $\PCO$ and its fragments are characterized by rather special classes of inequalities, the languages with the strict tensor turn out to be associated to the class of arbitrary linear inequalities; thus, the introduction of this operator, which is as far a we know absent from the literature on causation, seems to be supported by the criterion of mathematical elegance.  
Finally, we established that conditional probability statements of the forms $\Pr(\alpha \mid \beta)\leq \Pr(\gamma \mid \delta)$ and $\Pr(\alpha \mid \beta)\leq \Pr(\gamma)$   
are not in general expressible in $\PCOs$, and separated $\PCO$ (and $\PCOs$) from its extension $\PCO^\omega$ with infinitary disjunctions.

Analogous to the folklore result that the logic $L_{\infty\omega}$ can define all classes of finite structures, it was shown in \cite{BarSan2024} that the same holds for $\PCO^\omega$ with respect to all classes of causal multiteams that are closed under \emph{rescaling} and have the \emph{empty multiteam property}. 
This essentially amounts to full expressive power if the probabilities are given a non-frequentist interpretation. 
While any logic that is expressively complete in this sense is uncountable, 
it is an interesting task to identify more expressive finitary languages.
We describe some future directions of research:
\begin{itemize}
    \item Can (conditional) probabilistic independence atoms be expressed in $\PCO$? We conjecture the negative in line with \cite[Proposition 26]{HanVir2022}, which establishes that they are not expressible in $\mathrm{FO}(\approx)$, the \emph{probabilistic inclusion logic} of \cite{HanHirKonKulVir2019} (although the proof in \cite{HanVir2022} relies on the use of quantifiers).
    \item How can our results be extended to cover infinite signatures? Here one might need to extend the languages with quantifiers ranging over data values.
\item Our characterizations cover only logics that express linear properties of data. Can we generalize our results if some natural source of multiplication, such as conditional probabilistic independence or the conditional comparison atoms, are added to the logics? It was shown by Hannula et al. \cite{HKBV2020}  that the so-called \emph{probabilistic independence logic} is equiexpressive with a variant of existential second-order logic that has access to addition and multiplication of reals.
\item Finally, a promising direction for future work would be to study temporal aspects of causal inference (see e.g., \cite{Kleinberg11}) via (probabilistic) temporal logics by generalising the \emph{temporal team semantics} introduced by Krebs et al. \cite{KrebsMV018} and further developed by Gutsfeld et al. \cite{GutsfeldMOV22}.\looseness=-1
\end{itemize}

We conclude by pointing out the formal similarity of our work with some results obtained for first-order logics with probabilistic dependencies, such as the aforementioned language $\mathrm{FO}(\approx)$. Such languages do not formalize causation, and yet we can conjecture that $\PCO$ might be embeddable in $\mathrm{FO}(\approx)$ (similarly as the language $\CO$ is embedded into first-order logic in \cite{BarGal2022}). This idea is supported by a result of Hannula and Virtema (\cite{HanVir2022}) that establishes that definability in $\mathrm{FO}(\approx)$ can be reformulated in linear programming. It is however unknown which exact fragment of linear programming corresponds (in the sense of our Table \ref{table:exp}) to the language $\mathrm{FO}(\approx)$; such a characterization would give precise limits to the possibility of embedding results.

\section*{Acknowledgements}
The research of the first author was conducted under the Academy of Finland grant n. 349803. The second author was supported by the DFG grant VI 1045/1-1.
We thank Milo Orlich for kindly providing us with Figure \ref{figure: defining semipolytopes in PCOs}. 




\bibliographystyle{plainurl}
\bibliography{iilogics}

\appendix

\section{Transferring results from causal team semantics}\label{Appendix: transferring results}
In this section we provide tools for transferring results from causal team semantics to causal multiteam semantics, in particular Lemma \ref{lemma: transfer} and Theorem \ref{thm: phif}, which are needed to prove some of the results of this paper. We formulate Lemma \ref{lemma: transfer} and related results for language $\CO$, although the methods extend to other logics considered in the literature (e.g. extensions of $\CO$ with dependence atoms and/or with the global disjunction).
 
For most purposes, we can think of a causal team (of signature $\sigma$) as a pair $(T^-,\F)$, where $T^-$ is a team instead of a multiteam (i.e., a set of assignments on $Dom$ instead of $Dom \cup \{Key\}$), satisfying the conditions given in Definition \ref{def: causal multiteam}. In previous papers, the definition of causal teams differs in the sense that one can have distinct function components which differ from each other just in the sense that some of their functions have different sets of dummy argument. Such function components were dubbed to be \emph{equivalent} in \cite{BarYan2022}, and a notion of causal team equivalence was derived from it. In our framework, equivalence just coincides with equality. Lemma \ref{lemma: transfer} below will show that these differences of detail do not constitute a serious obstacle for transferring results from the framework of \cite{BarYan2022} to our causal multiteam semantics (in particular, for deriving Theorem \ref{thm:COchar} from \cite[Theorem 4.4]{BarYan2022}). 

We write $\models$ for the satisfaction relation over causal multiteams, and $\models^{ct}$ for the satisfaction relation over causal teams (when there is a need to make the distinction). The constructions $T^\alpha$ and $T_{\SET X = \SET x}$ are defined analogously as in the multiteam case.
The satisfaction clauses of the causal team semantics for language $\CO$, as given in previous literature, are formally identical to those of causal multiteam semantics, with the exception that the symbol $\lor$ is interpreted as \emph{lax} tensor, that is:

\begin{center}
$T\models^{ct} \psi\lor \chi$ if there are two causal subteams $T_1,T_2$ of $T$ such that $T_1^-\cup T_2^- = T^-$, $T_1\models^{ct} \psi$ and $T_2\models^{ct} \chi$
\end{center}
i.e. the causal subteams $T_1,T_2$ are not required to be disjoint.




Given a multiteam $T^-$ of signature $\sigma=(\dom,\ran)$, there is a corresponding team $\Team(T^-):= \{s_{\upharpoonright \dom} \mid s\in T^-\}$ of signature $\sigma$. More generally, given a causal multiteam  $T = (T^-,\F)$ of signature $\sigma$, there is a corresponding causal team $\Team(T) = (\Team(T^-), \F)$.

The following key lemma allows to translate results between causal team semantics and causal \emph{multi}team semantics. 

\begin{lemma}\label{lemma: transfer}
Let $T$ be a causal multiteam of signature $\sigma$, and $\varphi\in\CO$. Then:
\[
T\models\varphi \iff \mathrm{\Team}(T)\models^{ct} \varphi
\]
\end{lemma}

\noindent Towards the proof of this result, we introduce three further lemmas.

\begin{lemma}\label{lemma: cm equals ct on singletons}
Let $\sigma=(\dom,\ran)$ be a signature. Let $s\in\B_\sigma$. 
 Then for all $n\in\mathbb N$, 
 $\F$ function component over $\dom$ and 
 $\varphi\in \CO_{\sigma}$ we have:
\[
(\{s(n/Key)\},  \F)\models \varphi \iff (\{s\},  \F)\models^{ct} \varphi.
\]
\end{lemma}

\begin{proof}
A straightforward proof by induction on $\varphi$, hinging on the fact that the variable $Key$ is not used in $\CO_{\sigma}$. In the case for $\lor$, we use the fact that, over a singleton causal (multi)team $T$, $T\models \psi_1\lor\psi_2\iff T\models \psi_1$ or $T\models \psi_2$. \qed
\end{proof}

\begin{lemma}\label{lemma: transfer alpha}
Let $T$ be a causal multiteam of signature $\sigma$, and $\alpha\in\CO_{\sigma}$. Then $\Team(T^\alpha)= \Team(T)^\alpha$.
\end{lemma}

\begin{proof} 
Write $T=(T^-,\F)$. 
Then
\begin{align*}
\Team(T^\alpha) &= \Team((\{s\in T^- \mid (\{s\},\F)\models\alpha\},\F))\\
& = (\{s_{\upharpoonright \dom} \mid s \in T^- \text{ and } (\{s\},\F)\models\alpha\},\F)\\
& = (\{s_{\upharpoonright \dom} \mid s \in T^- \text{ and } (\{s_{\upharpoonright \dom}\},\F)\models^{ct}\alpha\},\F)\\
& = (\{s_{\upharpoonright \dom} \mid s \in T^-\},\F)^\alpha\\
& = \Team((T^-,\F))^\alpha\\
& = \Team(T)^\alpha,
\end{align*}
where in the third equality we used Lemma \ref{lemma: cm equals ct on singletons}.
\qed
\end{proof}

\begin{lemma}\label{lemma: transfer X x}
Let $T$ be a causal multiteam of signature $\sigma=(\dom,\ran)$, $\SET X\subseteq \dom$ and $\SET x\in \ran(\SET X)$. Then $\Team(T_{\SET X = \SET x}) = \Team(T)_{\SET X = \SET x}$.
\end{lemma}

\begin{proof}
Write $T=(T^-,\F)$. 
Then 
\begin{align*}
\Team(T_{\SET X = \SET x}) 
&= (\{t_{\upharpoonright \dom} \mid t\in T_{\SET X = \SET x}^-\}, \F_{\SET X = \SET x})\\
&= (\{t_{\upharpoonright \dom} \mid \exists s\in T^- \text{ such that } t = s_{\SET X = \SET x}^\F\}, \F_{\SET X = \SET x})\\
&= (\{(s_{\upharpoonright \dom})_{\SET X = \SET x}^\F \mid s\in T^-\}, \F_{\SET X = \SET x})\\
&= (\{s_{\upharpoonright \dom} \mid s\in T^-\}, \F)_{\SET X = \SET x}\\
&= \Team(T)_{\SET X = \SET x},
\end{align*}
where in the third equality we used the fact that $(s_{\SET X = \SET x}^\F)_{\upharpoonright \dom} = (s_{\upharpoonright \dom})_{\SET X = \SET x}^\F$. \qed
\end{proof}

\noindent We can now prove the key lemma. 

\vspace{7pt}

\noindent \emph{Proof of Lemma \ref{lemma: transfer}.}
By induction on $\varphi$. The case for $\land$ 
is straightforward.

\begin{itemize}
\item Base case: $\varphi$ is $X=x$. (The case for $X\neq x$ is completely analogous.)

$\Rightarrow$) Suppose $T\models X=x$. Then, for every $s\in T^-$, $s(X)=x$; therefore $s_{\upharpoonright \dom}(X)=x$. Since every assignment in $\Team(T^-)$ is of the form $s_{\upharpoonright \dom}$ for some $s\in T^-$, this amounts to saying that, for all $t\in \Team(T^-), t(X)=x$; i.e., $\Team(T)\models^{ct} X=x$.

$\Leftarrow$) Suppose $\Team(T)\models^{ct} X=x$, i.e., for all $t\in \Team(T)^-, t(X)=x$. For each $n\in \mathbb N$, write $s^n_t$ for the assignment $t(n/Key)$. Then $s^n_t(X)=x$. Now $T^-\subseteq \{s^n_t \mid t\in \Team(T)^-, n \in \mathbb N\}$; therefore $T\models X=x$.





\item Case $\varphi$ is $\psi_1\lor \psi_2$.

$\Rightarrow$) If $T=(T^-,\F)\models\psi_1\lor \psi_2$ then there are disjoint multiteams $S_1,S_2\subseteq T^-$ such that 
 $S_1 \cup S_2 = T^-$ and $(S_i,\F)\models \psi_i$. Now define teams $S_i^*:=\{ s_{\upharpoonright \dom} \mid s\in S_i\}= \Team(S_i)$. By the inductive hypothesis, $(S_i^*,\F)\models^{ct} \psi_i$. Furthermore, if $t\in \Team(T^-)$, then $t=s_{\upharpoonright \dom}$ for an $s$ that belongs to either $S_1$ or $S_2$; so $t\in S_1^*$ or  $t\in S_2^*$. That is, $S_1^* \cup S_2^* = \Team(T^-)$. Thus $\Team(T)\models^{ct} \psi_1\lor \psi_2$. 

$\Leftarrow$) Assume $\Team(T)\models^{ct} \psi_1\lor \psi_2$. Then there are teams $S_1^*,S_2^*\subseteq \Team(T^-)$ such that $S_1^* \cup S_2^* = \Team(T^-)$ and $(S_i^*,\F)\models^{ct} \psi_i$. Now define the disjoint multiteams $S_1:=\{s\in T^- \mid s_{\upharpoonright \dom}\in S_1^*\}$ and 
$S_2:=T^-\setminus S_1$.
 Obviously then 
  $S_1 \cup S_2 = T^-$ and $S_1^* = \Team(S_1)$, and by the inductive hypothesis $(S_1,\F)\models\psi_1$. 


Let us show that $S_2 \subseteq \{s\in T^- \mid s_{\upharpoonright \dom}\in S_2^*\}$.  If $s\in S_2$, then  $s\notin S_1$; thus, $s_{\upharpoonright \dom}\notin S_1^*$. But then, since $s_{\upharpoonright \dom}\in \Team(T^-)$, $s_{\upharpoonright \dom}\in S_2^*$, as claimed. 

Now, since $S_2^*\models \psi_2$, by inductive hypothesis $\{s\in T^- \mid s_{\upharpoonright \dom}\in S_2^*\}\models \psi_2$; and since $S_2 \subseteq \{s\in T^- \mid s_{\upharpoonright \dom}\in S_2^*\}$, by the downward closure of $\CO$ we have $S_2\models \psi_2$. Thus, in conclusion, $T\models \psi_1 \lor \psi_2$.




\item Case $\varphi$ is $\alpha\supset \psi$.

 $T\models \alpha\supset \psi \iff T^\alpha\models\psi \iff$ (by the inductive hypothesis) $\Team(T^\alpha)\models^{ct} \psi \iff$ (by Lemma \ref{lemma: transfer alpha}) $\Team(T)^\alpha\models^{ct} \psi \iff \Team(T)\models^{ct} \alpha\supset\psi$.

\item Case $\varphi$ is $\SET X = \SET x \cf \psi$.

$T\models\SET X = \SET x \cf \psi \iff T_{\SET X = \SET x}\models  \psi \iff$ (by the inductive hypothesis) $\Team(T_{\SET X = \SET x})\models^{ct}  \psi\iff$ (by Lemma \ref{lemma: transfer X x}) $\Team(T)_{\SET X = \SET x}\models^{ct}  \psi\iff \Team(T)\models^{ct} \SET X = \SET x \cf \psi$. 
\qed
\end{itemize}


\noindent This result can be easily extended to formulae with $\sqcup$ or $\dep{\SET X}{Y}$. Notice also that the case for $\lor$ may fail in languages that are not downward closed. 
We will give two applications of Lemma \ref{lemma: transfer}, which were already mentioned in the main text. First, we will show that there are formulae $\Phi^\F$ that characterize the property of ``having function component $\F$'' in causal multiteam semantics. Secondly, we will semantically characterize $\CO$ over causal multiteams.

Let us begin by considering the issue of the $\Phi^\F$ formulae. 
The paper \cite{BarYan2022} used a slightly different semantics, in which there may exist causal functions $F_V, G_V$ that only differ for their set of dummy arguments. 
 For example, the functions $F_V(U,X):= U +X -X$ and $G_V(U,Z):= U +Z -Z$ have different argument variables, but they produce the same values of $V$ for each given value of $U$; $X$ is a dummy argument for $F_X$ and $Z$ is a dummy argument for $G_X$. In such a case, $F_V$ and $G_V$ coincide over $\PA_V^F \cap \PA_V^G$, while the variables $\PA_V^F \setminus \PA_V^G$ are dummy arguments for $F_V$, and $\PA_V^G \setminus \PA_V^F$ are dummy arguments for $G_V$. When this happens we say that $F_V$ and $G_V$ are \textbf{similar}, and we write $F_V\sim G_V$. Notice furthermore that  each such function $F_V$ (possibly with dummy argument)  is similar to a (unique) \textbf{minimal canonical representative} $f_V$ - a function with no dummy arguments; and also similar to a (unique) \textbf{maximal canonical representative} $\F_V$ - a function whose arguments are all the variables in $\dom\setminus\{V\}$. The latter are just the kinds of functions we defined earlier in the main text.

The notion of similarity is then extended to function components as follows. We write $F,G$ for  function components in the sense of \cite{BarYan2022}. 
Write $\en(F)$ for the set of endogenous variables of $F$.\footnote{By the conventions of this paper, if a variable is in $\en(F)$ then it is generated by a \emph{non-constant} causal function.}
We say that $F$ and $G$ are similar ($F\sim G$) if $\en(F) = \en(G)$ and, for all $V\in \en(F)$, $F_V\sim G_V$.\footnote{This definition looks simpler than that in \cite{BarYan2022} due to our convention that causal functions must be non-constant.} 
 Finally, we say that two causal (multi)teams $S=(S^-,F), T=(T^-,G)$ are \textbf{equivalent} ($S\approx T$) iff $S^-=T^-$ and $F\sim G$.

We denote as $\mathbb{F}_\sigma$ the (finite) set of all function components of signature $\sigma$. 
It was shown in \cite{BarYan2022} that, for each $\F\in\mathbb{F}_\sigma$, there is a formula $\Phi^\F \in \CO_\sigma$ (that is actually also in $\PCO_\sigma$) that characterizes the property of having function component $\F$ in the causal team semantics setting.\footnote{More precisely, in \cite{BarYan2022} this formula characterized the function component only up to similarity. 
In our framework, 
the similarity relation collapses to identity.} This characterization still holds in our framework, in the sense that, for every \emph{nonempty} causal multiteam $T = (T^-,\G)$ of signature $\sigma$:
\[
T\models \Phi^\F \iff \G = \F.
\]
%
Recall that $\SET W$ is an ordered list of all the variables in $\dom$ and that we fixed an enumeration $s_1,\dots,s_n$ of the assignments in $\B_\sigma$. The formula $\Phi^\F$, slightly adapted to the conventions of this paper\footnote{The original formula mentioned the set of constant causal functions, which are not allowed here. Furthermore, it had to refer explicitly to the parent set of causal functions.}, is:
\[
\Phi^\F: \bigwedge_{V\in \en(\F)} \eta_\sigma(V) \land \bigwedge_{V\notin \en(\mathcal F)  
} \xi_\sigma(V)
\]
where
\[
\eta_\sigma(V): \bigwedge_{\SET w\in \ran(\SET W_V)}(\SET W_V = \SET w \cf V = \F_V(\SET w))
\]
and
\[
\xi_\sigma(V): \bigwedge_{\substack{\SET w\in \ran(\SET W_V) \\ v \in \ran(V)}} V=v \supset (\SET W_V=\SET w \cf V=v).
\]

\begin{theorem}\label{thm: phif}
Let $T=(T^-,\G)$ be a \emph{nonempty} causal multiteam. Then,
\[
T\models \Phi^\F \iff \G=\F.
\]
\end{theorem}

\begin{proof}
In \cite{BarYan2022}, Theorem 3.4, it it proved that if $S = (S^-,G)$ is a causal team of signature $\sigma$, then $S\models \Phi^F  \iff G \sim F$. In our case, $T = (T^-,\G)\models \Phi^\F$ iff $Team(T)\models^{ct} \Phi^\F$ (by lemma \ref{lemma: transfer}); thus, by the theorem in \cite{BarYan2022}, iff $\G\sim\F$. But since $\F$ and $\G$ are both maximal canonical representatives, it must be $\F=\G$. Vice versa, trivially $\F=\G$ implies $\G\sim\F$, and then we can use the same equivalences as before, in the opposite direction. \qed
%
\end{proof}

As mentioned above, the language $\CO$ and its extensions including $\sqcup$ and dependence atoms received semantic characterizations in causal team semantics \cite{BarYan2022}. Lemma \ref{lemma: transfer}  allows us to convert these results (for non-probabilistic languages) into characterizations in causal multiteam semantics. We consider here only the case of language $\CO$, whose expressive power will be seen to be characterized by the property of flatness alone (Theorem \ref{thm:COchar}).



In general, a key property of non-probabilistic languages 
 is support-closedness:  
\begin{itemize}
\item $\K$ is \textbf{support-closed} if, whenever $T\in\K$ and $\Team(T)=\Team(S)$, then $S\in \K$.
\end{itemize}
 
\noindent A moment of thought shows that support-closedness logically follows from flatness, and thus we did not need to mention the former in the semantic characterization of language $\CO$. It would naturally appear in characterization theorems for more general languages. 
For example, it can be proved that the language that extends $\CO$ with dependence atoms (denoted $\COD$ in previous literature) is characterized by the nonempty multiteam property, downward closure and support closure; we omit the similar proof.

Let us work towards a proof of Theorem \ref{thm:COchar}. 

Given a class $\K$ of causal multiteams, we define the class of causal teams 
$\Team_\sigma^\approx(\K) = \{S \text{ of signature $\sigma$ } \mid \exists T\in\K : S\approx \Team(T)\}$. Let us see what kind of closure properties are preserved when passing from $\K$ to 
$\Team_\sigma^\approx(\K)$.

\begin{lemma}\label{lemma: transferable closure properties}
Let $\K$ be a class of causal multiteams of a common signature $\sigma$.

1) If $K$ is flat, then $\Team_\sigma^\approx(\K)$ is flat.

2) If $K$ is downward closed and support closed, then  $\Team_\sigma^\approx(\K)$ is downward closed.  
 
\end{lemma}


\begin{proof}
Let us consider first the case of downward closure. Let $T=(T^-,G)\in \Team_\sigma^\approx(\K)$ and $S$ a causal submultiteam of it. Since $T\in \Team_\sigma^\approx(\K)$, there is a $T'=((T')^-,\G)\in \K$ such that $T\approx \Team(T')$; then  $G\sim\G$. By the support-closure of $\K$, we can assume without loss of generality that there is a bijection between $T^-$ and $(T')^-$, i.e. each assignment in $(T')^-$ is of the form $t_s = s(n/Key)$ (for distinct values of $n$), where $s\in T^-$.
Define $S' = ((S')^-,\G)$, where $(S')^- = \{t_s \mid s\in S^-\}$. Then $S' \leq T'$; since $T'\in \K$ and $\K$ is downward closed, $S'\in\K$. However, clearly $S\approx \Team(S')$. Thus $S\in  \Team_\sigma^\approx(\K)$.

Now flatness. 
 Let $T = (T^-,F)$ be a causal team of signature $\sigma$, and assume that, for all $s\in T^-$, $(\{s\}, F) \in \Team_\sigma^\approx(\K)$. But then, using also the support-closure of $\K$ (which follows from flatness), there is a (unique!) $\F$, made of the maximal canonical representatives of the functions given by $F$, such that $(\{s\}, \F)\in\Team_\sigma^\approx(\K)$ for each $s\in T^-$. By definition of $\Team_\sigma^\approx(\K)$, there are numbers $n_s\in \mathbb N$ such that $(s(n_s/Key),\F)\in\K$ for each $s\in T^-$. Since $\K$ is support closed, we can assume that the $n_s$ are distinct; thus, $S:= (\{s(n_s/Key) \mid s\in T^-\}, \F)$ is a causal multiteam. By the flatness of $\K$, $S\in\K$. But clearly $\Team(S) = (T^-,\F)$, so $(T^-,\F)\in \Team_\sigma^\approx(\K)$. Finally, since by definition $\Team_\sigma^\approx(\K)$ is closed under $\approx$, $T=(T^-, F)\in \Team_\sigma^\approx(\K)$. 
 The opposite direction immediately follows from the downward closure case. \qed
%
\end{proof}

We can now prove the correctness of the characterization of language $\CO_{\sigma}$.


\COchar*

\begin{proof}
$\Rightarrow$) 
This is just Theorem \ref{thm: CO flatness}.

$\Leftarrow$) Assume $\K$ satisfies flatness; as observed before, then, $\K$ is support closed.  Let $\Team_\sigma^\approx(\K) = \{S \text{ of signature $\sigma$ } \mid \exists T\in\K : S\approx \Team(T)\}$ as before. 
By Lemma \ref{lemma: transferable closure properties},   $\Team_\sigma^\approx(\K)$ satisfies the causal team version of flatness. Furthermore, by definition $\Team_\sigma^\approx(\K)$ is closed under the equivalence relation $\approx$. But then, by the characterization of $\CO_{\sigma}$ over causal teams (Theorem 4.4 of \cite{BarYan2022}), there is a formula $\varphi\in\CO_{\sigma}$ which defines $\Team_\sigma^\approx(\K)$ over causal teams of signature $\sigma$. We show that the same formula defines $\K$ over causal multiteams of signature $\sigma$.

Let $T\in\K$. Then $\Team(T)\in \Team_\sigma^\approx(\K)$. So $\Team(T)\models^{ct} \varphi$. By Lemma \ref{lemma: transfer}, then, $T\models \varphi$.

In the opposite direction, let $T$ be a causal multiteam such that $T\models \varphi$. By Lemma \ref{lemma: transfer}, $\Team(T)\models^{ct} \varphi$. But then $\Team(T)\in \Team_\sigma^\approx(\K)$. Since $\Team(T)\in \Team_\sigma^\approx(\K)$, there is an $S\in \K$ such that $\Team(S) \approx \Team(T)$. This entails, in particular, that the function components of $\Team(S)$ and $\Team(T)$ are similar. But since they are also the function components of $S$, resp. $T$, the fact that they are similar just means that they are identical. Thus $\Team(S) = \Team(T)$.  
 Then, by the support-closure of $\K$, $T\in\K$. \qed
\end{proof}





\end{document}